\documentclass{article}
\usepackage{ijcai26}

\usepackage{times}
\usepackage{soul}
\usepackage{url}
\usepackage[hidelinks]{hyperref}
\usepackage[utf8]{inputenc}
\usepackage[small]{caption}
\usepackage{graphicx}
\usepackage{amsmath}
\usepackage{amsthm}
\usepackage{booktabs}
\usepackage{algorithm}
\usepackage{algorithmic}
\usepackage{amsfonts,amssymb}
\usepackage{tcolorbox}
\tcbset{size=small,top=1mm,bottom=1mm,middle=1mm}
\usepackage[switch]{lineno}
\usepackage[T1]{fontenc}

\newtheorem{theorem}{Theorem}[section]

\newtheorem{proposition}[theorem]{Proposition}
\newtheorem{lemma}[theorem]{Lemma}

\newtheorem{claim}{Claim}

\theoremstyle{definition}
\newtheorem{definition}[theorem]{Definition}

\newtheorem{remark}[theorem]{Remark}

\newcommand{\util}{\textsc{Util}}

\newcommand{\sat}{\mathrm{sat}}
\newcommand{\JR}{\text{JR}}
\newcommand{\PJR}{\text{PJR}}
\newcommand{\EJR}{\text{EJR}}
\newcommand{\EJRplus}{\text{EJR+}}

\title{The Price of Proportional Representation in Temporal Voting}
\author{
    Nicholas Teh
    \affiliations
    University of Oxford
    \emails
    nicholas.teh@cs.ox.ac.uk
}

\begin{document}

\maketitle

\begin{abstract}
    We study proportional representation in the temporal voting model, where collective decisions are made repeatedly over time over a fixed horizon. Prior work has extensively investigated how proportional representation axioms from multiwinner voting (e.g., justified representation (JR) and its variants) can be adapted, satisfied, and verified in this setting. However, much less is understood about their interaction with social welfare. In this work, we quantify the efficiency cost of enforcing  proportionality. We formalize the welfare-proportionality tension via the worst-case ratio between the maximum achievable utilitarian welfare and the maximum welfare attainable subject to a proportionality axiom. We show that imposing proportional representation in the temporal setting can incur a growing, yet sublinear, welfare loss as the number of voters or rounds increases. We further identify a clean separation among axioms: for JR, the welfare loss diminishes as the time horizon grows and vanishes asymptotically, whereas for stronger axioms this conflict persists even with many rounds. Moreover, we prove that welfare maximization under each axiom is NP-complete and APX-hard, even under static preferences and bounded-degree approvals, and provide fixed-parameter algorithms under several natural structural parameters.
\end{abstract}

\section{Introduction}
Many collective decisions are inherently multi-stage: rather than selecting a single outcome once, an institution commits to a sequence of choices over a horizon.
In modern ML-driven systems, this occurs whenever a platform allocates scarce \emph{exposure} repeatedly: for instance, selecting a featured creator each day, a highlighted job posting each week, or a promoted product category each campaign cycle.
A large body of work in information retrieval and recommender systems studies exactly this kind of repeated exposure allocation, where fairness is evaluated \emph{across many outputs} (often called \emph{amortized} or \emph{long-run} fairness) and must be balanced against utility or engagement \cite{biega_equity_2018,singh_fairness_2018,pitoura_vldb_2022}.
These settings also need not be online in the algorithmic sense: platforms commonly plan or optimize slates and exposure schedules in batch using historical logs, forecasts, or business constraints.

We study such problems through the lens of the \emph{temporal voting} model, a framework for long-term collective decision-making \cite{bulteau2021jrperpetual,chandak2024proportional,elkind2025verifying,phillips2025strengthening}.
An instance consists of a finite horizon of $\ell$ rounds together with an approval profile for each round.\footnote{
We study the \emph{offline} planning problem: the entire $\ell$-round instance (all round profiles) is given as input, and the rule outputs the full length-$\ell$ sequence at once.}
In each round the rule selects one alternative, and alternatives may be selected multiple times.
This repeated nature of selection creates a fundamental tension:
selecting, in each round, an alternative that has maximum support will guarantee high utilitarian welfare.
However, repeatedly using such a greedy rule can systematically exclude minorities; for instance, even a sizeable segment of the population may see none of its approved alternatives selected over the entire horizon.
Accordingly, the natural fairness question is whether we can achieve some kind of representation \emph{across time}.

The concept of \emph{proportional representation} is particularly relevant here.
Recent work adapts classic proportionality axioms from approval-based multiwinner voting---in particular, justified representation (JR), proportional justified representation (PJR), and extended justified representation (EJR)---to the temporal setting \cite{bulteau2021jrperpetual,chandak2024proportional}.
These axioms ensure that any sufficiently large and cohesive group obtains its ``fair share'' of representation when representation is evaluated over the whole horizon.
More recently, strengthened variants such as EJR+ were proposed to further tighten proportionality guarantees over time \cite{phillips2025strengthening}.

However, while these axioms provide compelling group-representation guarantees, they may force the outcome to deviate from the welfare-maximizing optimum.
Importantly, this tension arises even under static preferences and offline optimization, because proportionality imposes global, group-based constraints that couple the choices across rounds.

This raises two basic questions that are central both in theory and practice:
(1) How costly is enforcing proportional representation in temporal voting?
(2) Can we compute a high-welfare outcome that satisfies JR/PJR/EJR/EJR+?

To address the first question, we adopt the ``price of fairness'' perspective and study the \emph{price of proportional representation}: the worst-case ratio between the maximum utilitarian welfare achievable without constraints and the maximum utilitarian welfare achievable subject to a proportionality axiom.
This viewpoint provides a crisp, quantitative answer to the welfare-fairness tension: it identifies when proportionality is essentially ``free'' and when it provably imposes a substantial welfare penalty.

To address the second question, we study the computational complexity of optimizing welfare subject to these proportionality axioms.
While proportionality constraints are conceptually appealing, they implicitly encode global combinatorial structure (i.e., which cohesive groups must be represented and how their representation must be distributed over time) and it is not clear a priori whether this structure admits efficient optimization, even under restrictive preference assumptions.

\subsection{Our Contributions}
We study four proportionality axioms $\{\JR,\PJR,\EJR,\EJRplus\}$ in the temporal voting setting, and analyze their compatibility with utilitarian social welfare.

In Section~\ref{sec:price}, we \emph{quantify} the welfare-proportionality tension via a \emph{price of proportional representation} framework. 
We establish a general lower bound showing that enforcing any of $\JR/\PJR/\EJR/\EJRplus$ can incur a growing (yet sublinear) loss of $\Omega(\sqrt{\ell})$ in the worst case, where $\ell$ is the number of rounds. 
We also prove a universal upper bound of $n$ (the number of agents) on the price. 
For $\JR$, we go further and obtain a tight, horizon-sensitive guarantee: the price decreases with~$\ell$ and becomes asymptotically negligible when $\ell\gg n$ (approaching $1$). 
In contrast, for the stronger axioms $\PJR/\EJR/\EJRplus$, the price can remain $\Omega(\sqrt{n})$ even when $\ell=\Theta(n)$, giving us a clear separation between JR and its stronger counterparts.

In Section~\ref{sec:computational}, we show that maximizing utilitarian welfare subject to any of $\JR/\PJR/\EJR/\EJRplus$ is NP-complete, and that the corresponding optimization problem is APX-hard, even under static preferences and constant approval-degree bounds. This demonstrates that the computational barrier is intrinsic to the group-based proportionality constraints rather than being dependent on time-varying preferences.

Finally, in Section~\ref{sec:parameterized}, we complement the hardness results by identifying practically meaningful regimes where optimal solutions can be computed efficiently. 
Under static preferences, we obtain fixed-parameter tractability with respect to the number of candidates~$m$. 
More generally, we develop algorithms that exploit temporal structure, including FPT results parameterized by the number of voter types (and related temporal-compression parameters), showing that repeated preference patterns across rounds enable tractable welfare maximization under temporal proportionality.

Overall, our results provide the first systematic study of the efficiency cost and computational compatibility of temporal proportional representation axioms with social welfare.

\subsection{Related Work}
We briefly survey the most relevant strands of the literature and emphasize works that are directly related to our work.

\paragraph{Temporal voting.}
Proportional representation in temporal voting was first studied by Bulteau et al.~\shortcite{bulteau2021jrperpetual}.
They adapted two proportional representation axioms (JR and PJR) from multiwinner voting to the temporal setting (with analyses for both static and changing preferences).
Chandak et al.~\shortcite{chandak2024proportional} subsequently extended this analysis to the stronger EJR, and showed that EJR exists in the temporal setting.
Elkind et al.~\shortcite{elkind2025verifying} built on this and study the complexity of checking whether a given outcome satisfies JR/PJR/EJR.
Recently, Phillips et al.~\shortcite{phillips2025strengthening} examined various ways of strengthening these concepts, and showed that EJR+ and FJR (both strengthenings of EJR) are always satisfiable. We refer the reader to the survey by Elkind et al.~\shortcite{elkind2024temporalsurvey} for an overview of this area.

In an adjacent line of work, Elkind et al.~\shortcite{elkind2024temporalelections} study temporal elections with a focus on welfare and incentive-compatibility.  
They investigate the computational complexity of welfare maximization and its compatibility with strategyproofness.
They also adapt an individual proportionality notion (primarily studied in fair division and public decision-making), and study its compatibility and tradeoffs with utilitarian and egalitarian welfare objectives. This distinction is important for our results. The individual proportionality
notion studied by Elkind et al.~\shortcite{elkind2024temporalelections} is vacuous in some of our hardness constructions.
Subsequently, Elkind et al.~\shortcite{elkind2025temporalchores} consider analogous welfare optimization questions for negatively valued candidates, where proportionality notions are not well-defined.
Our work differs in that we focus on \emph{group-based} proportional representation axioms, which capture collective fairness guarantees that are fundamentally different from individual proportionality (where relevant, we will point out the connection and highlight the differences).

Another similar (but fundamentally distinct) body of work
studies the perpetual voting model \cite{lackner2020perpetual,lackner2023proportionalPV}. 
This model differs from temporal voting in that decisions explicitly depend on the history of past rounds, giving it a more inherently online flavor. Other temporal voting models examine different kinds of intertemporal structure: for example, Zech et al.~\shortcite{zech2024multiwinnerchange} study multi-stage multiwinner elections in which each new committee should remain close to the preceding one. These works are complementary to ours: rather than controlling change between outcomes or designing an online rule, we ask how imposing group-based proportionality
constraints over a fixed horizon affects utilitarian welfare and
computational tractability.

\paragraph{Proportionality in multiwinner voting.}
The proportionality axioms we study are inspired by the well-developed literature on proportional representation in approval-based multiwinner voting. 
In that setting, justified representation (JR) and extended justified representation (EJR) were introduced by Aziz et al.~\shortcite{aziz2017jr}, and proportional justified representation (PJR) was later proposed as an intermediate axiom between JR and EJR by Sanchez-Fernandez et al.~\shortcite{sf2017pjr}. 
Subsequently, Brill and Peters~\shortcite{brill2023robust} introduced EJR+, a strengthening of EJR that preserves desirable practical properties such as polynomial-time verifiability and satisfaction by polynomial-time computable rules.

Temporal variants of these axioms adapt the same group-protection intuition, but the introduction of rounds fundamentally changes both the structure of the constraints (e.g., how representation is required to accrue over time) and the way proportionality interacts with efficiency. 
Moreover, unlike in multiwinner voting, the same alternative may be selected repeatedly across rounds, which further alters the behavior of proportionality axioms and the design space of rules.
Our work advances this line of research by quantifying the welfare-proportionality trade-off for temporal variants of the JR family, and by showing that the length of the time horizon affects this trade-off in qualitatively different ways depending on the strength of the axiom.
Equivalently, even under static preferences, a temporal outcome is a multiset
of $\ell$ candidate occurrences rather than a committee of $\ell$ distinct
or interchangeable winners: multiplicities matter because selecting the same
candidate repeatedly consumes rounds and repeatedly satisfies the same voters.

Proportionality has also been studied in the context of \emph{public decision-making} (simultaneous aggregation over multiple issues), both without constraints \cite{skowron2022proppublic} and with feasibility constraints \cite{chingoma2025proportionalityconstrained} through adaptations of proportional representation axioms from multiwinner voting, as well as in even more general models of feasibility constraints \cite{masarik2024generalizedtheory}.
A related multi-issue perspective is taken by Alouf-Heffetz et al.~\shortcite{alouf2022better}, who study how an external intervention that reduces agents' uncertainty can improve collective decisions under issue-by-issue majority voting. This line of work focuses on uncertainty in public decision-making, highlighting how informational and structural constraints can alter the welfare properties of collective outcomes.

\paragraph{Welfare-proportionality tradeoffs.}
A complementary line of work studies proportionality constraints through a quantitative lens, often phrased as a ``price of fairness'' or ``price of proportionality'' measure.
In the context of multiwinner voting, Elkind et al.~\shortcite{elkind2024price} initiate a systematic study of the welfare and coverage losses incurred by imposing JR and EJR constraints.
Related quantitative analyses adopt approximation-style guarantees to compare voting rules against welfare- or representation-oriented benchmarks \cite{LacknerSkowron2020}.
Similar welfare–representation tradeoffs have also been examined in participatory budgeting \cite{fairstein2022welfarerepPB}, and in designing multiwinner rules that simultaneously provide utilitarian and representation guarantees \cite{BrillPeters2024priceable}.

Our work extends this quantitative perspective to the temporal voting setting by defining and bounding the utilitarian price of proportionality.

\section{Preliminaries} \label{sec:prelims}
For every natural number $k \in \mathbb{N}$,
we let $[k] = \{1, 2, \dots, k\}$.
A \emph{temporal election} is a tuple
$E = (P, N, \ell, (\mathbf{s}_i)_{i \in N})$, where
$P$ is the set of $m$ {\em candidates},
$N$ is the set of $n$ \emph{voters},
$\ell$ is the number of \emph{rounds}, and for each $i \in N$, $\mathbf{s}_i = (s_{i,1}, s_{i,2},\dots, s_{i,\ell})$, where $s_{i,r} \subseteq P$ is the \emph{approval set} of voter $i$ in round $r$, which consists of candidates that $i$ approves in round $r$.\footnote{Preferences are said to be \emph{static} if for every voter $i \in N$ there exists a fixed approval set
$A_i \subseteq P$ such that $s_{i,r} = A_i$ for every round $r \in [\ell]$.
}
We denote the set of all temporal elections by $\mathcal{E}$ and the set of temporal elections whereby in every round every voter approves at least one candidate by $\mathcal{E}_{\geq 1}$ (which we call \emph{complete} elections).
We say that voters from a subset $S \subseteq N$ \emph{agree} in a round $r \in [\ell]$ if there exists a candidate that they all approve in this round, i.e., $\bigcap_{i \in S} s_{i,r} \neq \varnothing$.

An \emph{outcome} of a temporal election $E$ is a sequence $\mathbf{o} = (o_1, o_2, \dots, o_\ell) \in P^\ell$ 
of $\ell$ candidates such that for every $r \in [\ell]$, candidate $o_r \in P$ is chosen in round $r$.
A candidate may be selected multiple times, i.e., it may be the case that $o_t = o_{r'}$ for some $r \neq r'$.
For a subset of rounds $R \subseteq [\ell]$ and an outcome $\mathbf{o}$, we write $\mathbf{o}_R = (o_r)_{r \in R}$ to denote the \emph{suboutcome} with respect to $R$.
The \emph{satisfaction} of a subset of voters $S \subseteq N$ from a suboutcome $\mathbf{o}_R$ is $\sat_S(\mathbf{o}_R)=\big|\{r\in R :\ o_r \in \bigcup_{i \in S} s_{i,r}\} \big|$, i.e., the number of rounds in $R$ in which the selected candidate is approved by at least one voter from $S$.
If $S = \{i\}$ for some $i \in N$, then we simply write $\sat_i(\mathbf{o}_R)$.

\subsection{Proportional Representation Axioms}
We now provide the definitions of the proportionality axioms we consider, which were originally introduced in prior work on temporal voting \cite{bulteau2021jrperpetual,chandak2024proportional,elkind2025verifying}.

\begin{definition}[JR/PJR/EJR] \label{defn:jr_pjr_ejr}
    Given a temporal election $E = (P, N, \ell, (\mathbf{s}_i)_{i \in N})$ and an outcome $\mathbf{o}$, if, for each $t \in [\ell]$ and every nonempty subset of voters $S \subseteq N$ that agrees in a size-$t$ subset of rounds, it holds that
    \begin{itemize}
        \item $\sat_S(\mathbf{o}) \geq \min(1,\lfloor t \cdot |S|/n \rfloor)$, then $\mathbf{o}$ provides \emph{justified representation} (JR),
        \item $\sat_S(\mathbf{o}) \geq \lfloor t \cdot |S|/n \rfloor$, then $\mathbf{o}$ provides \emph{proportional justified representation} (PJR),
        \item $\sat_i(\mathbf{o}) \geq \lfloor t \cdot |S|/n \rfloor$ for some $i \in S$, then $\mathbf{o}$ provides \emph{extended justified representation} (EJR).
    \end{itemize}
\end{definition}
While JR, PJR, and EJR impose proportionality guarantees based on the total number of rounds on which a group agrees, they do not ensure that such agreement translates into representation in \emph{any particular round}.
In particular, a group may be sufficiently large and cohesive to warrant proportional representation overall, yet still be unrepresented in a round where all its members agree. To capture this stronger notion of temporal proportionality, Phillips et al.~\shortcite{phillips2025strengthening} proposed extended justified representation\,+ (EJR+),\footnote{This notion is based on the analogous strengthening of EJR in the multiwinner voting setting \cite{brill2023robust}.} which requires that sufficiently cohesive groups either obtain proportional satisfaction overall or receive representation in every round on which they unanimously agree.
It is defined as follows.

\begin{definition}[EJR+] \label{defn:ejrplus}
    Given a temporal election $E = (P, N, \ell, (\mathbf{s}_i)_{i \in N})$ and $\sigma \in [n]$, $\tau \in [\ell]$, we say that a nonempty subset of voters $S \subseteq N$ is \emph{$(\sigma, \tau)$-cohesive} if there exists a set of $\tau$ rounds $R \subseteq [\ell]$ and a suboutcome $\mathbf{o}_R = (o_r)_{r \in R}$ such that for each $r \in R$, at least $\sigma$ voters in $S$ approve $o_r$.
    
    We say that an outcome $\mathbf{o}$ provides \emph{extended justified representation\,+ (EJR+)} if for all $\sigma \in [n]$, $\tau \in [\ell]$, every $(\sigma, \tau)$-cohesive nonempty subset of voters $S \subseteq N$ and every round $r \in [\ell]$ with $\bigcap_{i \in S} s_{i,r} \neq \varnothing$, it holds that  (i) $\sat_i(\mathbf{o}) \geq \lfloor \tau \cdot \sigma/n \rfloor$ for some $i \in S$, or (ii) $o_r \in \bigcap_{i \in S} s_{i,r}$.
\end{definition}

Phillips et al.~\shortcite{phillips2025strengthening} proposed two strengthenings of EJR that are always satisfiable: EJR+ and \emph{full justified representation} (FJR).
The latter is inspired by an analogous strengthening studied in the participatory budgeting setting—a generalization of multiwinner voting that allows items to have variable costs \cite{peters2021fjr}.
In this work, we focus on EJR+ rather than FJR, as EJR+ is better established in the multiwinner voting literature and preserves key practical properties of its multiwinner analogue---most notably polynomial-time verifiability and satisfaction by a polynomial-time computable rule.
In contrast, while FJR is always satisfiable, the known constructive approach relies on a procedure whose running time is not polynomial, and Phillips et al.~\shortcite{phillips2025strengthening} observe that it remains open whether FJR outcomes can be computed in polynomial time.

\subsection{Price of Proportional Representation}
In this work, we focus on utilitarian social welfare, which is the canonical efficiency benchmark in approval-based voting and is standard in quantitative welfare-proportionality tradeoff analyses
\cite{LacknerSkowron2020,elkind2024price,elkind2024temporalelections}.
We view utilitarian welfare as the most informative baseline for isolating the efficiency cost of proportional representation:
unlike other (e.g., egalitarian or Nash) objectives, which already encode inequality aversion, utilitarian welfare measures aggregate efficiency directly.
Moreover, utilitarian welfare is additively separable across rounds, making the unconstrained optimum efficiently computable in our model and enabling transparent worst-case  comparisons.

Given a temporal election $E = (P, N, \ell, (\mathbf{s}_i)_{i \in N})$ and an $\mathbf{o}=(o_1,\dots,o_\ell)\in P^\ell$, the \emph{utilitarian welfare} of $\mathbf{o}$ under $E$ is the total number of approving voter-round pairs:
\begin{equation*}
    \util_E(\mathbf{o}) := \sum_{i\in N} \sat_i(\mathbf{o}) =\sum_{r=1}^{\ell}\bigl|\{ i\in N : o_r \in s_{i,r}\}\bigr|.
\end{equation*}
For notational simplicity, we omit the subscript $E$ when it is obvious from context.
More generally, one could consider an arbitrary welfare function
$W_E : P^\ell \to \mathbb{R}_{\ge 0}$.

Let $\Phi$ be a property of the outcome, e.g.,
$\Phi \in \{\JR, \PJR, \EJR, \EJR+\}$.
Fix integers $n,\ell \ge 1$, and let $\mathcal E^{n,\ell}$ be the
class of temporal elections with $n$ voters and $\ell$ rounds; define
$\mathcal E^{n,\ell}_{\ge 1}$ analogously for complete elections.
For an election $E \in \mathcal E^{n,\ell}$, define the unconstrained
optimum by $W^*(E) := \max_{\mathbf{o}\in P^\ell} W_E(\mathbf{o})$, and the $\Phi$-constrained optimum by $W^\Phi(E) := \max\{W_E(\mathbf{o}) : \mathbf{o} \in P^\ell \text{ and } \mathbf{o} \text{ provides } \Phi\}$.
For a fixed election $E$ with $W^\Phi(E)>0$, let $\rho_{W,\Phi}(E) := \frac{W^*(E)}{W^\Phi(E)}$.
For a class $\mathcal C \subseteq \mathcal E^{n,\ell}$, the price of
$\Phi$ with respect to $W$ over $\mathcal C$ is $\mathcal P_W(\Phi;\mathcal C)
:= \sup_{E\in \mathcal C: W^\Phi(E)>0} \rho_{W,\Phi}(E)$.
In Section~\ref{sec:price}, we take $\mathcal C=\mathcal E^{n,\ell}_{\ge 1}$, since the
price bounds there are for complete elections. To lighten notation, we let $\mathcal P_W^{n,\ell}(\Phi) := \mathcal P_W(\Phi;\mathcal E^{n,\ell}_{\ge 1})$ when the underlying class is clear.

Then, we will study the following decision problem.
\begin{tcolorbox}[title=$\Phi$-\textsc{UTIL}]
    \textbf{Input}: A temporal election $E$ and an integer threshold $B$.

    \textbf{Question}: Does there exist an outcome $\mathbf{o} \in P^\ell$ such that $\mathbf{o}$ provides $\Phi$ and $\util(\mathbf{o}) \geq B$?
\end{tcolorbox}
We also consider the associated optimization problem.
\begin{tcolorbox}[title=$\Phi$-\textsc{MaxUTIL}]
    \textbf{Input}: A temporal election $E$.

    \textbf{Task}: Compute an outcome $\mathbf{o}\in P^\ell$  such that $\mathbf{o}$ provides $\Phi$ and $\util(\mathbf{o})$ is maximized.
\end{tcolorbox}

\section{Price of Proportional Representation} \label{sec:price}
We begin by quantifying the welfare cost of enforcing proportional representation in temporal elections.

Here, we make two simple assumptions that are needed to ensure that the price measure is well-defined and meaningful.
First, we restrict attention to complete elections $E\in\mathcal{E}_{\ge 1}$, i.e., $s_{i,r}\neq\varnothing$ for all $i\in N$
and all $r\in[\ell]$. Without this, the ratio can be inflated by
trivial ``empty approval'' rounds in which no candidate is approved, and the measure would be uninformative.
Second, we focus on the setting where $\ell \geq n$.
To see why, our proportionality axioms ``assign representation'' in indivisible units of rounds via terms $\lfloor t \cdot |S| /n \rfloor$.
When $\ell < n$, even a voter who has a nonempty approval set at every round has $\lfloor \ell/n\rfloor = 0$, and JR/PJR/EJR impose no per-voter guarantee and can be satisfied while ignoring some voters in every round.
Moreover, this is unavoidable: in complete elections, one can have disjoint approvals so that each round can satisfy at most one voter, making ``everyone is represented at least once'' infeasible unless $\ell \geq n$.
This is also the minimal non-degenerate setting in which the proportionality axioms impose positive individual entitlements and have some bite; and thus, the resulting price isolates the meaningful welfare-proportionality tension.

We begin with a general lower bound.

\begin{proposition}\label{prop:sqrt-lb}
    Fix any $\ell \ge 1$ and $\Phi \in \{\JR, \PJR, \EJR, \EJR+ \}$.
There exists a complete temporal election $E\in \mathcal E^{\ell,\ell}_{\ge 1}$ such that $\rho_{\util,\Phi}(E) \ge \frac{1}{2}\sqrt{\ell}$.
Consequently, $\mathcal P_{\util}^{\ell,\ell}(\Phi) \ge \frac{1}{2}\sqrt{\ell} = \frac{1}{2}\sqrt n$.
\end{proposition}
Intuitively, the construction creates a ``core'' block of about $\sqrt{\ell}$ voters who share a consistently popular candidate, alongside many voters whose approvals are essentially private.
The welfare-maximizing outcome repeatedly satisfies the popular block, but any $\Phi$-feasible outcome must spend many rounds on low-support candidates to ensure some baseline representation for the private voters, producing a $\Theta(\sqrt{\ell})$ welfare gap.

Next, we provide a universal upper bound, which serves more as a coarse sanity check showing imposing proportionality cannot blow up welfare by more than a factor $n$.

\begin{proposition}\label{prop:trivial-upper}
    Fix any $n,\ell \ge 1$, any $\Phi \in \{ \JR, \PJR, \EJR, \EJR+ \}$, and any complete temporal election $E\in \mathcal E^{n,\ell}_{\ge 1}$ with $\util^{\Phi}(E)>0$. Then $\rho_{\util,\Phi}(E) \le n$. Consequently, $\mathcal P_{\util}^{n,\ell}(\Phi) \le n$.
\end{proposition}
The proof combines two simple facts: (i) welfare in any round is at most $n$, and $\util^*(E)\le n\ell$; (ii) for $E \in \mathcal{E}_{\geq 1}$, any $\Phi$-feasible outcome can be modified
(without affecting feasibility) so that every round selects a candidate approved by at least one voter, guaranteeing welfare at least $\ell$.

For JR we can substantially strengthen Proposition~\ref{prop:trivial-upper}; crucially, the bound
improves with number of rounds.

\begin{theorem}\label{thm:jr-ub}
    Fix $n,\ell$ with $\ell \ge n$. Then, for every
complete temporal election $E\in \mathcal E^{n,\ell}_{\ge 1}$, $\rho_{\util,\JR}(E)
\le \frac{\ell}{\ell-n+2\sqrt n-1}$.
Consequently, $\mathcal P_{\util}^{n,\ell}(\JR)
\le \frac{\ell}{\ell-n+2\sqrt n-1}$.
In particular, for every fixed $\varepsilon>0$, if
$\ell=(1+\varepsilon)n$, then $\mathcal P_{\util}^{n,\ell}(\JR)
\le \frac{1+\varepsilon}{\varepsilon+o(1)}$,
and if $\ell/n\to\infty$, then
$\mathcal P_{\util}^{n,\ell}(\JR)\to 1$.
\end{theorem}
At a high level, the proof is constructive: it identifies $n$ ``cheap'' rounds and uses them to ensure that every voter is represented at least once, while selecting per-round welfare maximizers in the remaining rounds. An averaging argument then shows that even on the reserved $n$ rounds, one can retain a $\Theta(1/\sqrt{n})$ fraction of the local welfare optimum.

We complement Theorem~\ref{thm:jr-ub} with a matching lower bound, showing tightness up to lower-order
terms.

\begin{proposition}\label{prop:jr-tight}
    Fix $n,\ell$ with $\ell \ge n$. There exists a complete temporal election $E\in \mathcal E^{n,\ell}_{\ge 1}$ such that $\rho_{\util,\JR}(E) \ge \frac{\ell}{\ell-n+2\sqrt n- \mathcal{O}(1)}$.
    Consequently, $\mathcal P_{\util}^{n,\ell}(\JR) \ge \frac{\ell}{\ell-n+2\sqrt n- \mathcal{O}(1)}$.
\end{proposition}

\begin{remark}
    Elkind et al.~\shortcite{elkind2024temporalelections} study temporal elections under an \emph{individual} proportionality
    axiom PROP, which requires each voter $i$ to be satisfied in at least $\lfloor \mu_i/n\rfloor$ rounds,     where $\mu_i$ is the number of rounds in which $i$ approves at least one candidate (so $\mu_i=\ell$ on complete elections).
    When $\lfloor \ell/n\rfloor=1$ (in particular, when $\ell=n$), PROP reduces to the requirement that
    every voter is satisfied at least once; on complete elections, this condition coincides with JR.
    In this setting, Elkind et al.~\cite[Thm.~5.6]{elkind2024temporalelections} determine the worst-case utilitarian price exactly as $n/(2\sqrt{n}-1)$; their lower-bound instance corresponds to the $\ell=n$ special case of
    Proposition~\ref{prop:jr-tight}.
    Our focus is different: we analyze group-based axioms and explicitly track how the welfare loss varies
    with the horizon $\ell$; Theorem~\ref{thm:jr-ub} and Proposition~\ref{prop:jr-tight} give a tight
    $\ell$-sensitive bound for JR, showing that its price decreases with $\ell$ and tends to $1$ when
    $\ell\gg n$.
\end{remark}
Finally, we make explicit the qualitative separation suggested by Theorem~\ref{thm:jr-ub}: increasing the time horizon can make JR cheap, but it need not reduce the cost of stronger axioms.

\begin{theorem}\label{thm:long-horizon-sep}
    Fix an integer $a \ge 2$ and set $\ell = an$. Then $P^{n,an}_{\util}(\JR) \le \frac{a}{a-1}+o(1)$ as $n\to\infty$. However, for each $\Phi \in \{\PJR,\EJR,\EJR+\}$, $P^{n,an}_{\util}(\Phi)=\Omega(\sqrt n)$.
\end{theorem}

The results in this section show that enforcing proportional
representation in temporal elections can incur a provable sublinear welfare loss. 
For JR, this loss decreases with $\ell$ and vanishes as $\ell$ becomes large relative to $n$ (Theorem~\ref{thm:jr-ub}), whereas for stronger axioms (PJR/EJR/EJR+) the cost can remain $\Omega(\sqrt{n})$ even when $\ell$ is a constant-factor multiple of $n$ (Theorem~\ref{thm:long-horizon-sep}).

\section{Computational Intractability Results} \label{sec:computational}
The previous section quantified the welfare loss that can arise from enforcing proportionality in temporal elections.
We now turn to the algorithmic question: given a temporal election, can we compute a high welfare outcome that
satisfies a chosen proportionality axiom?
Formally, we study the decision problem \textsc{$\Phi$-UTIL} and its optimization variant \textsc{$\Phi$-MaxUTIL} (as defined in Section~\ref{sec:prelims}) for $\Phi \in \{\JR, \PJR, \EJR, \EJRplus\}$.

Our key findings are: for each $\Phi \in \{\JR, \PJR, \EJR, \EJRplus\}$, maximizing utilitarian welfare subject to $\Phi$ is NP-hard, and even APX-hard, under very strong structural restrictions.
In particular, our reductions use static preferences (approval sets do not vary across rounds) and constant approval degree bounds (each candidate is approved by a constant number of voters), so that feasibility verification is straightforward; the hardness comes from the combinatorial property of simultaneously meeting many group-based representation constraints while retaining high welfare.

We begin with two simple structural observations that let us treat
$\JR/\PJR/\EJR$ (and, on static instances, also $\EJRplus$) in a unified way.
We then prove NP-completeness of \textsc{$\Phi$-Util} and strengthen this to APX-hardness for \textsc{$\Phi$-MaxUtil}
via a gap-preserving reduction.

Our NP-hardness construction then forces every \emph{positive} proportional entitlement to be exactly~$1$.
Intuitively, whenever a cohesive group is large/frequent enough to trigger a proportionality requirement at all, it is only entitled to \emph{one} unit of representation over the whole horizon.
In this setting, the distinctions between $\JR$, $\PJR$, and $\EJR$ disappear: all three axioms reduce to the same proportionality condition.

\begin{lemma}\label{lem:alpha01-collapse}
    Fix a temporal election $E=(P,N,\ell,(\mathbf{s}_i)_{i\in N})$.
    Assume that for every integer $t>0$ and every nonempty voter group $S\subseteq N$
    that agrees in a size-$t$ subset of rounds, we have $\left\lfloor t\cdot |S| / n \right\rfloor \le 1$.
    Then for any outcome $\mathbf{o}$, the following are equivalent: $\mathbf{o}$ satisfies JR, $\mathbf{o}$ satisfies PJR, and $\mathbf{o}$ satisfies EJR.
\end{lemma}

Lemma~\ref{lem:alpha01-collapse} lets a single reduction establish hardness for all of $\JR/\PJR/\EJR$, provided we ensure the stated condition.
Our constructions achieve this in a clean, verifiable way: we enforce a constant bound $k$ on the number of approvers per candidate, which implies that any agreeing group
$S$ must be contained in the approver set of some candidate and hence has size at most $k$.
By choosing $n$ sufficiently large relative to $\ell$ and~$k$, we guarantee
$\lfloor t\cdot |S|/n\rfloor \le \lfloor \ell \cdot k/n\rfloor \le 1$ for all relevant $(S,t)$.

Next we handle $\EJRplus$.
While $\EJRplus$ is strictly stronger than $\EJR$ in general, under static preferences it collapses back to $\EJR$:
if a group ever unanimously agrees, it agrees in \emph{every} round, so $\EJR$ already forces enough aggregate
satisfaction to satisfy the $\EJRplus$ disjunction via~(i).

\begin{lemma}\label{lem:static_ejr_ejr+}
    For any temporal election with static preferences, an outcome satisfies $\EJRplus$ iff it satisfies $\EJR$.
\end{lemma}

Together, Lemmas~\ref{lem:alpha01-collapse} and~\ref{lem:static_ejr_ejr+} imply that, on the static, ``low demand'' instances used in our reductions, all axioms $\Phi \in \{\JR,\PJR,\EJR,\EJRplus\}$ impose the same feasibility constraints
(up to notational differences).

We now show that \textsc{$\Phi$-UTIL} is NP-complete even
under very strong restrictions.
The reduction is designed to isolate the algorithmic difficulty to the group-based representation constraints:
preferences are static, and each candidate is approved by only constantly many voters, so the structure of potentially binding groups is simple; nevertheless, selecting a high welfare, feasible outcome encodes an NP-complete covering problem.

\begin{theorem} \label{thm:intractable_NPc}
    Fix any $\Phi \in \{\JR,\PJR,\EJR,\EJRplus\}$.
    Then, $\Phi$-\textsc{UTIL} is NP-complete, even when voters have static preferences and each candidate is approved by at most $4$ voters. 
\end{theorem}

The above result rules out polynomial-time exact optimization unless $\mathrm{P}=\mathrm{NP}$.
One might still hope for a PTAS for \textsc{$\Phi$-MaxUtil}.
We rule this out by proving APX-hardness under similarly strong restrictions.

\begin{theorem} \label{thm:intractable_APXh}
    Fix any $\Phi \in \{\text{JR, PJR, EJR, EJR+}\}$.
    Then $\Phi$-\textsc{MaxUTIL} is APX-hard, even when voters have static preferences, each voter approves at most $3$ candidates and each candidate is approved by at most $8$ voters.
\end{theorem}

    It is worth contrasting our computational results with those of Elkind et al.~\shortcite{elkind2024temporalelections}.
    They show NP-hardness of deciding whether there exists an \emph{individually} proportional (PROP) outcome that also maximizes utilitarian welfare.
    Our focus is different: we study \emph{group-based} temporal proportionality axioms $\JR/\PJR/\EJR/\EJRplus$, which impose fundamentally different constraints (and are in general incomparable with PROP).
    Moreover, our hardness reductions already go through under static preferences, so the intractability is not due to time-varying approvals.
    Finally, both reductions operate in settings where $\lfloor \ell/n\rfloor = 0$, so PROP is vacuous, yet
    $\JR/\PJR/\EJR/\EJRplus$ remain nontrivial due to cohesive groups—highlighting that the difficulty is inherent
    to \emph{group} representation over time.

\section{Parameterized Complexity Results} \label{sec:parameterized}
The intractability results in the previous section hold even under static preferences, and rule out polynomial-time algorithms or approximation schemes in the worst case.
Nevertheless, temporal elections in practice often exhibit additional structure.
In such settings, worst-case hardness need not preclude efficient algorithms.
Thus, we study the parameterized complexity of $\Phi$-\textsc{MaxUTIL} (see Section~\ref{sec:prelims}). 
Our goal is to identify natural structural parameters under which the problem becomes tractable, and to make explicit how the algorithmic difficulty depends on these parameters.
Throughout, we give exact algorithms (rather than approximation guarantees), thereby complementing the hardness results of Section~\ref{sec:computational}.

As a simple baseline, when the number of rounds $\ell$ is treated as a parameter, $\Phi$-\textsc{MaxUTIL} can be solved by brute-force enumeration of outcomes: we enumerate all sequences $\mathbf{o} \in P^\ell$, compute $\util(\mathbf{o})$, and retain the best $\Phi$-feasible outcome. 
Since verifying $\Phi$-feasibility of a given outcome $o$ is in XP with respect to $\ell$~\cite[Prop.~5.3]{elkind2025verifying}, enumerating all $m^\ell$ outcomes gives us an XP algorithm with respect to $\ell$: for some computable function $f$, the running time is $m^\ell \cdot (n+m)^{f(\ell)}$.
Thus, $\Phi$-MAXUTIL is in XP with respect to $\ell$. While this approach might be impractical for large $\ell$, it provides a useful reference point in settings where $\ell$ is constant/bounded.
We then turn to more meaningful structural parameters.

\subsection{Fixed $m$ under Static Preferences}
Section~\ref{sec:computational} shows that even under static preferences, $\Phi$-\textsc{UTIL} is NP-hard and $\Phi$-\textsc{MaxUTIL} is APX-hard for all $\Phi \in \{ \JR, \PJR, \EJR, \EJR+\}$; thus, tractability cannot be recovered merely by assuming static preferences. 
However, this structural assumption becomes algorithmically useful when combined with small parameters. 
In particular, if the number of candidates $m$ is fixed, an outcome is characterized (up to permuting rounds) by its candidate multiplicities (i.e., how many times each candidate is selected).
Formally, each voter $i\in N$ approves a fixed set $A_i\subseteq P$ in every round. Any outcome $\mathbf{o} \in P^\ell$ is therefore fully specified (up to permuting rounds) by its \emph{multiplicity vector} $\mathbf{x}\in\mathbb{Z}_{\ge 0}^{P}$ with $x_p := |\{r\in[\ell]: o_r = p\}|\quad\text{and}\quad \sum_{p\in P} x_p = \ell$.
Under $\mathbf{x}$, voter $i$'s satisfaction and utilitarian welfare are linear: $\sat_i(\mathbf{x})=\sum_{p\in A_i} x_p$ and 
$\util(\mathbf{x})=\sum_{p\in P} w(p)x_p$, where $w(p):=|\{i\in N: p\in A_i\}|.$
To compress the instance, we group identical approval sets into \emph{approval classes}
$\mathcal{C}:=\{A_i : i\in N\}$, and write $n_C := |\{i\in N : A_i=C\}|$ for each class
$C\in\mathcal{C}$. All voters in a class share the same satisfaction value $v_C(\mathbf{x}) := \sum_{p\in C} x_p$.

The remaining work is to encode the proportionality axiom as constraints over $\mathbf{x}$ (and a bounded number of auxiliary integer variables).
When the number of candidates $m$ is fixed, these ILPs have a number of integer variables bounded by a function of $m$ only (e.g., via at most $|\mathcal{C}|\le 2^m$ auxiliary variables), and thus are solvable in time $f(m)\cdot \mathrm{poly}(n,\ell)$ for ILPs in fixed dimension using the result of Lenstra Jr~\shortcite{lenstra1983integer}.

Although EJR implies PJR implies JR (and on static instances EJR+ is equivalent to EJR by Lemma~\ref{lem:static_ejr_ejr+}), we state separate optimization theorems because the objective is welfare maximization under constraints: weaker axioms enlarge the feasible set, so an outcome that maximizes welfare subject to EJR may be feasible for PJR and JR but not optimal for them. Moreover, the constraint structure differs qualitatively. 

This contrasts with the ``low demand'' setting used in Section~\ref{sec:computational} (Lemma~\ref{lem:alpha01-collapse}), where the three axioms coincide; here demands can exceed $1$, so the distinctions matter for optimization.

\begin{theorem}\label{thm:JR-FPT-static-m}
    Fix a temporal election $E=(P,N,\ell,(A_i)_{i\in N})$ with static
    preferences, and let $m=|P|$. Then $\JR$-MAXUTIL is
    solvable in time $f(m)\cdot \mathrm{poly}(n,m,\ell)$ for some
    computable function $f$. In particular, it is FPT with respect to
    $m$.
\end{theorem}

\begin{theorem}\label{thm:PJR-FPT-static-m}
Fix a temporal election with static preferences $A_1,\dots,A_n\subseteq P$ and $\ell$ rounds.
Then PJR\textsc{-MaxUTIL} is solvable in time $f(m)\cdot \mathrm{poly}(n,\ell)$ for some computable function $f$. In particular, it is FPT with respect to $m$.
\end{theorem}

\begin{theorem}\label{thm:fpt-m-candidates-ejr}
    Fix a temporal election with static preferences $A_1,\dots,A_n\subseteq P$ and $\ell$ rounds. 
    Then EJR\textsc{-MaxUtil} (and thus EJR+-\textsc{MaxUtil}) is solvable in time $f(m)\cdot \mathrm{poly}(n,\ell)$ for some computable function $f$. In particular, it is FPT with respect to $m$.
\end{theorem}

\subsection{Fixed Voter Types}
We now move beyond static preferences while retaining a common form of structure in applications: although approvals may vary over time, the electorate often consists of a small number of \emph{behavioral types} (e.g., user segments) whose members react similarly in every round. 
We show that this structure gives us fixed-parameter tractability for welfare maximization under temporal proportionality.

Formally, let $T$ be a set of voter types and let $\kappa := |T|$. The voters are partitioned into disjoint sets $(N_\theta)_{\theta\in T}$, where $n_\theta := |N_\theta|$ and $\sum_{\theta\in T} n_\theta = n$. For each round $r\in[\ell]$, voters of type $\theta$ share a common approval set $A_{\theta,r}\subseteq P$; i.e., for every $i\in N_\theta$ we have $s_{i,r}=A_{\theta,r}$.

This representation allows us to work at the type level. For an outcome $\mathbf{o}=(o_1,\dots,o_\ell)$ define the (common) satisfaction of type $\theta$ by $\mathrm{sat}_\theta(\mathbf{o}) := \left|\{r\in[\ell] : o_r \in A_{\theta,r}\}\right|$.
For any voter $i\in N_\theta$, we have $\sat_i(\mathbf{o})=\sat_\theta(\mathbf{o})$.
Moreover, utilitarian welfare can be written compactly as
$\util(\mathbf{o}) = \sum_{r=1}^\ell \sum_{\theta\in T :  o_r \in A_{\theta,r}} n_\theta$.
It will be convenient to encode the effect of choosing a candidate in round $r$ by the set of approving types. 
For each round $r$ and candidate $p\in P$, let $X_r(p) := \{\theta\in T : p \in A_{\theta,r}\}\subseteq T$, and $w(X) := \sum_{\theta\in X} n_\theta$.
Then selecting $p$ in round $r$ contributes $w(X_r(p))$ to welfare and increases $\mathrm{sat}_\theta(\cdot)$ by $1$ exactly for $\theta\in X_r(p)$. 
Importantly, for a fixed round $r$, the welfare and all proportionality-relevant effects depend only on the \emph{type pattern} $X_r(p)$, not on the identity of $p$. Thus in round $r$ we may compress candidates into the family $\mathcal{X}_r := \{X_r(p) : p\in P\}\subseteq 2^T$, 
keeping one representative candidate per pattern (so $|\mathcal{X}_r|\le 2^\kappa$).

We show that this structure leads to efficient algorithms even when preferences can vary across rounds.
For JR, we obtain a dynamic program whose running time is exponential only in the number of types $\kappa$. 

\begin{theorem} \label{thm:votertype_jr}
    JR-\textsc{MaxUTIL} is solvable in time $f(\kappa)\cdot \mathrm{poly}(n,m,\ell)$ for some computable function $f$. In particular, it is FPT with respect to $\kappa$.
\end{theorem}

For EJR we need to account for \emph{how many times} each type is satisfied, since demands scale with group size and cohesion frequency. 
Let $\gamma_U := \left|\left\{r\in[\ell] : \bigcap_{\theta\in U} A_{\theta,r}\neq\emptyset\right\}\right|$ be the number of rounds in which all types in $U$ ``agree'' (i.e., have a unanimously approved candidate in that round).
Then, define for each $U \subseteq T$ the type set $U$ induces an integer demand $d_U:=\left\lfloor \gamma_U\cdot w(U)/n\right\rfloor$, and $D := \max_{U\subseteq T, U \neq \varnothing} d_U$.
Note that $D\le \ell$ (since $d_U\le \gamma_U\le \ell$). 
Under voter types, EJR can be expressed purely at the type level: for every $U$ with $d_U>0$, at least one type $\theta\in U$ must achieve satisfaction at least $d_U$.
Then, we get fixed-parameter tractability with respect to $\kappa$ and this demand parameter $D$.

\begin{theorem}\label{thm:ejr-maxutil-types}
    EJR\textsc{-MaxUTIL} is solvable in time $f(\kappa,D)\cdot \mathrm{poly}(n,m,\ell)$ for some computable function $f$. In particular, it is FPT with respect to $\kappa + D$.
\end{theorem}
The parameter $D$ captures the ``largest proportional claim'' any cohesive union of types can generate; when no type set can demand many satisfied rounds (e.g., because unanimous agreement is rare), the state space remains small.

Since EJR implies PJR, the DP in Theorem~\ref{thm:ejr-maxutil-types} also outputs a PJR-feasible outcome. However, because PJR is strictly weaker than EJR in general, this outcome need not be welfare-optimal among all PJR-feasible outcomes.

The above DPs scale linearly with $\ell$. In many temporal settings, however, the approval pattern repeats: many rounds share the same ``profile'' across types (e.g., weekly cycles). We can exploit this further by grouping rounds into a small number of distinct profiles.
Two rounds $r,r'\in[\ell]$ have the same profile if $A_{\theta,r}=A_{\theta,r'}$ for every $\theta\in T$. 
Then, we get the following result.

\begin{proposition} \label{prop:votertype_pjr}
    On instances with $\kappa$ voter types and $q$ round profiles,
    PJR\textsc{-MaxUTIL} is solvable in time $f(\kappa,q)\cdot \mathrm{poly}(n,\ell,m)$ for a computable function $f$. 
    In particular, it is FPT with respect to $(\kappa,q)$.
\end{proposition}

\noindent
The same profile-based variable scheme can also give us ILP formulations (and hence FPT algorithms parameterized by $(\kappa,q)$) for \textsc{JR-MaxUTIL} and \textsc{EJR-MaxUTIL} by adapting the constraints to the corresponding feasibility conditions.

Section~4 showed that static preferences alone do not restore tractability. In contrast, bounding the diversity of \emph{temporal behavior}—via a small number of voter types and, optionally, a small number of round profiles—does lead to exact FPT algorithms for welfare maximization under temporal proportionality. This pinpoints temporal structure, rather than time-invariance per se, as a main driver of algorithmic tractability in temporal voting.
Together, these results show that although $\Phi$-\textsc{MaxUTIL} is intractable in full generality, it becomes efficiently solvable in a range of structured settings that naturally arise in temporal voting settings. This delineates a clear boundary between worst-case hardness and practically relevant tractability.

\section{Conclusion}
In this work, we investigated proportional representation in temporal voting
through a quantitative \emph{price of proportional representation} framework,
measuring the worst-case loss in utilitarian welfare incurred by enforcing
JR/PJR/EJR/EJR+. Our bounds make the welfare-proportionality tension in this
setting precise: enforcing proportionality can cause a growing, yet sublinear,
welfare loss, and the magnitude of this loss depends sharply on both the axiom
and the time horizon. In particular, the cost of JR vanishes on long horizons:
its price approaches $1$ as the number of rounds grows large relative to the
number of voters. By contrast, stronger notions such as PJR/EJR/EJR+ can still
incur an $\Omega(\sqrt n)$ loss even when $\ell$ is only a constant-factor
multiple of $n$.

On the algorithmic side, we showed that maximizing utilitarian welfare subject to these proportionality constraints is NP-complete, and APX-hard, even under structural restrictions such as static preferences and bounded approval degrees. 
At the same time, we identified structured regimes in which optimal outcomes can be computed efficiently via fixed-parameter methods, including parameterizations by the number of candidates, the number of voter types, and the number of distinct round profiles. 

Our work naturally gives rise to several promising avenues for further work.
One direction is to refine the quantitative bounds, particularly for stronger proportionality notions such as PJR/EJR/EJR+.
Such refinements would help clarify whether the observed separations are inherent or driven by worst-case constructions.
A second direction is to extend the analysis beyond utilitarian welfare. Studying the price of proportionality with respect to alternative objectives (such as egalitarian or Nash welfare) may uncover qualitatively different trade-offs and lead to new insights into how proportionality interacts with social welfare under richer normative criteria.

\bibliographystyle{named}
\bibliography{abb,bib}

\clearpage 
\appendix
\begin{center}
\Large
\textbf{Appendix}
\end{center}

\vspace{2mm}

\section{Further Related Work}

\paragraph{Temporal allocation and fair division.}
Our model is also related to temporal allocation and fair division problems, where a sequence of decisions must be evaluated not only round by round but also in terms of the cumulative guarantees it provides over time. In temporal slot assignment, Elkind et al.~\shortcite{elkind2022temporalslot} study assignment of slots and consider proportionality-type objectives across the whole horizon; this is close in spirit to our setting, since selecting an alternative in each round can be viewed as allocating a slot. 
However, in our model the chosen outcome is a public decision observed by all voters, and proportionality is imposed through group-based voting axioms rather than through assignment-specific fairness criteria. 
A related line of work in temporal fair division studies how fairness evolves when indivisible items arrive sequentially: Elkind et al.~\shortcite{elkind2025temporalfairdivision} analyze cumulative approximate envy-freeness, while Neoh et al.~\shortcite{neoh2025online} and Choo et al.~\shortcite{choo2025approxproponline} investigate online fair division with predictions, and approximate proportionality guarantees. These works share with ours the idea that temporal structure can make fairness and efficiency interact in nontrivial ways, but they typically concern private allocations and online arrival models. Finally, Lim et al.~\shortcite{lim2025} study repeated matching under a maximin objective, providing another example where repeated decisions create intertemporal fairness-efficiency tradeoffs. These papers motivate the broader temporal decision-making perspective, while our contribution is to quantify and compute the welfare cost of enforcing JR/PJR/EJR/EJR+ in the temporal voting model.

\section{Omitted Proofs in Section~\ref{sec:price}}

\subsection{Proof of Proposition~\ref{prop:sqrt-lb}}
    Fix $\ell \ge 1$ and let $n := \ell$ and $k := \lceil \sqrt{\ell}\rceil$.
    Partition the voters into two disjoint sets $C, D$ with $|C| = k$ and $|D| = n-k = \ell-k$.
    
    Define the candidate set
    \[
    P := \{z\}\ \uplus\ \{x_{i,r} : i \in D,\ r \in [\ell]\}\ \uplus\ \{y_{c,r} : c \in C,\ r \in [\ell]\},
    \]
    where all symbols denote distinct candidates.
    Define the approval sets as follows for each round $r \in [\ell]$:
    \begin{itemize}
        \item $s_{c,r} := \{z,\, y_{c,r}\}$ for each $c \in C$, and
        \item $s_{i,r} := \{x_{i,r}\}$ for each $i \in D$.
    \end{itemize}
    Note that $s_{i,r}\neq \varnothing$ for all $i\in N$ and all $r\in[\ell]$, so $E\in \mathcal{E}_{\ge 1}$.
    
    Fix any round $r\in[\ell]$. Candidate $z$ is approved by exactly the $k$ voters in $C$.
    Every other candidate is approved by \emph{at most} one voter in round $r$
    (in particular, $x_{i,r}$ is approved only by voter $i$, $y_{c,r}$ only by voter $c$,
    and candidates with mismatched indices may be approved by no one in that round).
    Thus, the maximum welfare contribution in round $r$ is $k$, attained by selecting $z$.
    Therefore the unconstrained optimal outcome is $\mathbf{o}^*=(z,\ldots,z)$ and
    \[
    \util^*(E)=k\ell.
    \]
    
    Let $\mathbf{o}$ be any outcome that provides $\Phi$.
    Note that a $\Phi$-feasible outcome exists: select $x_{i,r}$ for each $i\in D$ in a distinct round (thus satisfying every $i\in D$ once), and select $z$ in the remaining $k$ rounds (thus satisfying all $c\in C$ many times). Since the only cohesive groups are subsets of $C$ (and singletons in $D$), this outcome satisfies JR/PJR/EJR, and also EJR+.
    
    Fix a voter $i\in N$ and consider the singleton group $S:=\{i\}$.
    Since $E \in \mathcal{E}_{\geq 1}$, $S$ agrees in every round (since
    $\bigcap_{j\in S} s_{j,r}=s_{i,r}\neq \varnothing$ for all $r$), and thus agrees in a size-$\ell$ subset of rounds. 
    As $n=\ell$, we have
    \begin{equation*}
        \left\lfloor \frac{\ell\cdot |S|}{n}\right\rfloor=\left\lfloor \frac{\ell}{\ell}\right\rfloor = 1.
    \end{equation*}
    If $\Phi=\JR$, then $\sat_i(\mathbf{o})\geq \min\{1,\lfloor \ell/n\rfloor\}=1$.
    If $\Phi=\PJR$, then $\sat_i(\mathbf{o})\geq \lfloor \ell/n\rfloor=1$.
    If $\Phi=\EJR$, then $\sat_i(\mathbf{o})\ge 1$.
    
    Finally, suppose $\Phi=\EJRplus$. 
    Set $\sigma:=1$ and $\tau:=\ell$.
    For each round $r\in[\ell]$, pick an arbitrary candidate $p_r\in s_{i,r}$ (possible since $s_{i,r}\neq \varnothing$),
    let $R:=[\ell]$, and let $\tilde o_R:=(p_r)_{r\in R}$. Then $S$ is $(\sigma,\tau)$-cohesive.
    Since $\bigcap_{j\in S} s_{j,r}=s_{i,r}\neq \varnothing$ for all $r$, Definition~\ref{defn:ejrplus} applied to $S$ and any round $r$
    implies either
    (i) $\sat_i(\mathbf{o})\ge \left\lfloor \tau\sigma/n\right\rfloor=\left\lfloor \ell/\ell\right\rfloor=1$,
    or (ii) $o_r\in \bigcap_{j\in S} s_{j,r}=s_{i,r}$, so voter $i$ is satisfied in round $r$, and thus $\sat_i(\mathbf{o})\ge 1$.
    Thus in all cases $\sat_i(\mathbf{o})\ge 1$.
    
    Consequently, for every $\Phi\in\{\JR,\PJR,\EJR,\EJRplus\}$ and every $\Phi$-feasible outcome $\mathbf{o}$,
    \begin{equation}\label{eq:prop31-singleton}
        \sat_i(\mathbf{o})\ge 1 \quad \text{for all } i\in N.
    \end{equation}
    
    Now, let $r_z := \bigl|\{r\in[\ell] : o_r = z\}\bigr|$ be the number of rounds in which $z$ is selected.
    In each round with $o_r=z$, exactly the $k$ voters in $C$ approve $o_r$, so that round contributes $k$ to welfare.
    In every other round, the selected candidate is not $z$ and is approved by at most one voter in that round,
    so that round contributes at most $1$.
    Thus,
    \begin{equation*}
        \util(\mathbf{o})\leq r_z\cdot k + (\ell-r_z)\cdot 1 = \ell + r_z(k-1).
    \end{equation*}
    Now focus on the voters in $D$. In round $r$, voter $i\in D$ approves exactly the candidate $x_{i,r}$.
    Moreover, for any fixed $r$ the candidates $\{x_{i,r}: i\in D\}$ are distinct and each is approved by a different voter,
    so at most one voter in $D$ can be satisfied in a single round.
    By~\eqref{eq:prop31-singleton}, every voter in $D$ must be satisfied at least once, so we need at least $|D|=\ell-k$
    rounds in which some voter in $D$ is satisfied. In any round where a voter in $D$ is satisfied, we cannot have chosen $z$
    (since no voter in $D$ approves $z$). Therefore,
    \[
    \ell - r_z \ge \ell-k \implies r_z \le k.
    \]
    This gives us
    \[
    \util(\mathbf{o})\le \ell + k(k-1) = \ell - k + k^2.
    \]
    Since $\mathbf{o}$ was an arbitrary $\Phi$-feasible outcome, it follows that
    \[
    \util^\Phi(E)\le \ell - k + k^2.
    \]
    Therefore,
    \[
    \frac{\util^*(E)}{\util^\Phi(E)} \ge \frac{k\ell}{\ell-k+k^2}.
    \]
    Let $s:=\sqrt{\ell}$, so $\ell=s^2$ and $k=\lceil s\rceil$.
    We claim that $\ell-k+k^2\le 2ks$. Indeed,
    \[
    2ks-(\ell-k+k^2)=2ks-(s^2-k+k^2)=k-(s-k)^2.
    \]
    Since $k=\lceil s\rceil$, we have $0\le k-s<1$, thus, $(s-k)^2<1\le k$, so the right-hand side is nonnegative.
    Thus $\ell-k+k^2\le 2ks$, and consequently
    \[
    \frac{\util^*(E)}{\util^\Phi(E)} \ge \frac{k s^2}{2ks} = \frac{s}{2} = \frac{1}{2}\sqrt{\ell}.
    \]
    By definition, $\mathcal{P}_\util(\Phi)$ is the supremum of $\util^*(E)/\util^\Phi(E)$ over elections with $\util_\Phi(E)>0$,
    so the existence of this election implies $\mathcal{P}_\util(\Phi)\geq \frac{1}{2}\sqrt{\ell}$.

\subsection{Proof of Proposition~\ref{prop:trivial-upper}}
    Fix $\Phi \in \{\JR,\PJR,\EJR,\EJRplus\}$ and let
    $E=(P,N,\ell,(s_i)_{i\in N}) \in \mathcal{E}_{\ge 1}$.
    Assume $\util^{\Phi}(E)>0$, i.e., there exists at least one $\Phi$-feasible outcome.
    
    For any outcome $\mathbf{o}=(o_1,\dots,o_\ell)\in P^\ell$,
    \begin{equation*}
        \util(\mathbf{o}) = \sum_{t=1}^{\ell} \bigl|\{\, i\in N : o_t \in s_{i,t}\,\}\bigr| \leq \sum_{t=1}^{\ell} n
        = n\ell,
    \end{equation*}
    because in each round at most $n$ voters can approve the selected candidate. 
    Thus, $\util^*(E)\le n\ell$.
    
    Now, let $\mathbf{o}=(o_1,\dots,o_\ell)$ be any outcome satisfying $\Phi$ (it exists by assumption).
    We transform $\mathbf{o}$ by eliminating \emph{empty} rounds, i.e., rounds in which the selected candidate is approved by no voter.
    
    Formally, for each round $t\in[\ell]$ with
    \begin{equation*}
        \{ i\in N : o_t \in s_{i,t}\}=\varnothing,
    \end{equation*}
    pick an arbitrary voter $i_t\in N$ and then pick any candidate $p_t\in s_{i_t,t}$.
    Such a candidate exists because $E\in \mathcal{E}_{\ge 1}$, i.e., $s_{i,t}\neq\varnothing$ for all $i\in N$ and all $t\in[\ell]$.
    Replace $o_t$ by $p_t$ and keep all other rounds unchanged.
    Let $\widehat{\mathbf{o}}$ denote the outcome obtained after performing this replacement for every empty round.
    
    By construction, for every round $t\in[\ell]$ there exists at least one voter (namely $i_t$ if the round was modified, or some voter who already approved $o_t$ otherwise) who approves $\widehat{o}_t$.
    Therefore each round contributes at least $1$ to utilitarian welfare, and thus
    \begin{equation*}
        \util(\widehat{\mathbf{o}}) \geq \ell.
    \end{equation*}
    
    It suffices to argue that a \emph{single} replacement (in one empty round) cannot violate $\Phi$; then applying this argument sequentially gives us that $\widehat{\mathbf{o}}$ is still $\Phi$-feasible.
    
    So fix an empty round $t$ and let $\mathbf{o}'$ be the outcome obtained from $\mathbf{o}$ by replacing $o_t$ with $p_t$.
    Since the round was empty, $o_t\notin s_{i,t}$ for every voter $i\in N$. Thus, no voter loses satisfaction in round $t$ when we replace $o_t$ by $p_t$, and no other round changes. 
    Concretely, for every $i\in N$,
    \begin{equation*}
        \sat_i(\mathbf{o}') = \sat_i(\mathbf{o}) + \mathbf{1}[p_t\in s_{i,t}] \ge \sat_i(\mathbf{o}).
    \end{equation*}
    Similarly, for every voter group $S\subseteq N$, round $t$ contributed $0$ to $\sat_S(\mathbf{o})$ (because $o_t$ was in no approval set at round $t$), so
    \begin{equation*}
        \sat_S(\mathbf{o}') = \sat_S(\mathbf{o}) + \mathbf{1}\left[\exists i\in S:\ p_t\in s_{i,t}\right] \ge \sat_S(\mathbf{o}).
    \end{equation*}    
    
    We first prove the result for JR/PJR/EJR.
    All constraints in Definition~\ref{defn:jr_pjr_ejr} have the form $\sat_S(\cdot)\ge \text{(threshold)}$ (JR, PJR)
    or $\exists i\in S$ with $\sat_i(\cdot)\ge \text{(threshold)}$ (EJR), where the thresholds depend only on the instance (and on $S,t$), not on the outcome.
    Since every $\sat_S$ and every $\sat_i$ weakly increases when passing from $\mathbf{o}$ to $\mathbf{o}'$, every JR, PJR, and EJR
    constraint that held for $\mathbf{o}$ also holds for $\mathbf{o}'$.

    Next, we prove the result for EJR+.
    Fix any $\sigma\in[n]$, $\tau\in[\ell]$, any $(\sigma,\tau)$-cohesive nonempty $S\subseteq N$, and any round
    $r\in[\ell]$ with $\bigcap_{i\in S} s_{i,r}\neq \varnothing$. We show the EJR+ requirement remains satisfied after the replacement.
    If $r\neq t$, then the chosen candidate in round $r$ does not change (only round $t$ was modified), so (ii) is unchanged; and (i) can only become easier to satisfy since all $\sat_i$ weakly increase.
    If $r=t$, then the original chosen candidate $o_t$ was approved by no voter, and thus, $o_t\notin \bigcap_{i\in S} s_{i,t}$; therefore (ii) was false for $\mathbf{o}$ at this constraint, so the
    constraint had to be satisfied via (i), i.e., $\sat_i(\mathbf{o})\ge \lfloor\tau\sigma/n\rfloor$ for some $i\in S$.
    Because $\sat_i(\mathbf{o}')\geq \sat_i(\mathbf{o})$ for all $i$, (i) remains true for $\mathbf{o}'$ as well.
    Thus EJR+ is preserved by the replacement.
    
    Consequently, $\widehat{\mathbf{o}}$ satisfies $\Phi$, and therefore
    \begin{equation*}
        \util^{\Phi}(E) \ge \util(\widehat{\mathbf{o}}) \ge \ell.
    \end{equation*}
    
    Combining $\util^*(E)\leq n\ell$ with $\util^{\Phi}(E)\geq \ell$, we get that
    \begin{equation*}
        \frac{\util^*(E)}{\util^{\Phi}(E)} \leq \frac{n\ell}{\ell} = n,
    \end{equation*}
    as claimed.

\subsection{Proof of Theorem~\ref{thm:jr-ub}}
    Fix a complete temporal election
    $E=(P,N,\ell,(s_i)_{i\in N})\in \mathcal E^{n,\ell}_{\ge 1}$
    with $\ell\ge n$.
    For each round $t\in[\ell]$, define the maximum approval count
    \[
    a_t \;:=\; \max_{p\in P}\bigl|\{\,i\in N : p\in s_{i,t}\,\}\bigr|.
    \]
    Since utilitarian welfare is additive across rounds, an unconstrained utilitarian optimum can be obtained by choosing,
    in each round $t$, a candidate that attains $a_t$ approvals; thus
    \begin{equation*}
        \util^*(E)=\sum_{t=1}^{\ell} a_t.
    \end{equation*}
    Let $B\subseteq[\ell]$ be a set of $n$ rounds minimizing $\sum_{t\in B}a_t$ (equivalently, $B$ consists of the $n$
    smallest values among $\{a_t\}_{t\in[\ell]}$). Write $\bar B := [\ell]\setminus B$ and define
    \[
    \mathrm{OPT}_B := \sum_{t\in B} a_t,
    \mathrm{OPT}_{\bar B} := \sum_{t\in \bar B} a_t
    = \util^*(E)-\mathrm{OPT}_B.
    \]
    For any set of rounds $R\subseteq[\ell]$ and any suboutcome $\mathbf{o}_R=(o_t)_{t\in R}$, define the welfare contributed by $R$ as
    \[
    \util_R(o_R)\;:=\;\sum_{t\in R}\bigl|\{\,i\in N : o_t\in s_{i,t}\,\}\bigr|.
    \]
    
    \medskip
    \noindent\textbf{Claim.}
    There exists a suboutcome $\mathbf{o}_B$ on the rounds in $B$ such that every voter is satisfied at least once within $B$
    (i.e., $sat_i(\mathbf{o}_B)\ge 1$ for all $i\in N$) and
    \[
    \util_B(\mathbf{o}_B)\ge \frac{2\sqrt{n}-1}{n}\cdot \mathrm{OPT}_B.
    \]
    
    \medskip
    \noindent\emph{Proof of Claim.}
    Let $q:=\max_{t\in B} a_t$ and fix $t^*\in B$ with $a_{t^*}=q$.
    Let $T\subseteq B$ be any set of size $q$ (since $q \leq n$ since $a_t \leq n$ for all rounds, so such a set exists) consisting of $t^*$ and the other $q-1$ rounds in $B$ with
    largest $a_t$-values.
    For each $t\in T$, choose a candidate $p_t\in P$ attaining $a_t$ approvals in round $t$ and set $o_t:=p_t$.
    Let
    \[
    Q := \sum_{t\in T} a_t,
    \]
    which is exactly the welfare contributed by the rounds in $T$.
    
    In particular, in round $t^*$ the chosen candidate is approved by $q$ voters, so at least $q$ voters are already satisfied
    after fixing the choices on $T$. Thus, at most $n-q$ voters remain unsatisfied.
    Since $|B\setminus T|=n-q$ and $s_{i,t}\neq\varnothing$ for all $i,t$, we can choose the remaining entries $o_t$ for
    $t\in B\setminus T$ so that:
    \begin{itemize}
    \item each previously unsatisfied voter becomes satisfied at least once within $B$
    (assign each such voter injectively to a distinct round in $B\setminus T$ and pick a candidate she approves); and
    \item every round in $B\setminus T$ selects a candidate approved by at least one voter
    (if some rounds remain unassigned, pick any voter and a candidate she approves in that round).
    \end{itemize}
    Thus every voter is satisfied at least once within $B$, and every round in $B\setminus T$ contributes at least $1$ to welfare.
    Therefore
    \[
    \util_B(\mathbf{o}_B)\ge Q + (n-q).
    \]
    
    If $B\setminus T=\varnothing$, then $q=|B|=n$ and the construction achieves $\util_B(\mathbf{o}_B)=\mathrm{OPT}_B$,
    so assume $B\setminus T\neq\varnothing$ and define
    \[
    r := \max_{t\in B\setminus T} a_t \;\ge\; 1.
    \]
    Then
    \[
    \mathrm{OPT}_B = \sum_{t\in T} a_t + \sum_{t\in B\setminus T} a_t
    \le Q + (n-q)\,r.
    \]
    Moreover, since $T$ contains the $q$ rounds of $B$ with largest $a_t$-values, we have $a_t\ge r$ for all $t\in T$,
    so $Q\ge rq$. Combining these bounds gives us
    \begin{align*}
        \frac{\util_B(\mathbf{o}_B)}{\mathrm{OPT}_B} 
         \geq \frac{Q+n-q}{Q+(n-q)r} 
         & \ge \frac{rq+n-q}{rq+(n-q)r} \\
         & = \frac{rq+n-q}{nr},
    \end{align*}
    where the second inequality follows from the fact that $r \geq 1$, so the function $f(Q)=\frac{Q+n-q}{Q+(n-q)r}$ is nondecreasing in $Q$, and $Q\ge rq$, so $f(Q)\ge f(rq)$.
    
    Since $q\ge r$, we have $rq+n-q = n+q(r-1)\ge n+r(r-1)$, and thus
    \[
    \frac{\util_B(\mathbf{o}_B)}{\mathrm{OPT}_B}
    \ge \frac{n+r(r-1)}{nr}
    = \frac{r-1}{n}+\frac{1}{r}
    =\left(\frac{r}{n}+\frac{1}{r}\right)-\frac{1}{n}.
    \]
    By the AM-GM inequality, $\frac{r}{n}+\frac{1}{r}\ge 2/\sqrt{n}$, so
    \[
    \frac{\util_B(\mathbf{o}_B)}{\mathrm{OPT}_B}
    \ge \frac{2}{\sqrt{n}}-\frac{1}{n}
    = \frac{2\sqrt{n}-1}{n},
    \]
    proving the claim. $ \hfill \square$
    
    \medskip
    Now construct a (full time horizon) outcome $\mathbf{o}$ as follows:
    \begin{itemize}
    \item for each $t\in \bar B$, choose a candidate attaining $a_t$ approvals in round $t$;
    \item for each $t\in B$, use the choice from $o_B$.
    \end{itemize}
    Then
    \begin{align*}
        \util(\mathbf{o}) & = \mathrm{OPT}_{\bar B}+\util_B(\mathbf{o}_B) \\
        & \ge \util^*(E)-\mathrm{OPT}_B +\frac{2\sqrt{n}-1}{n}\cdot \mathrm{OPT}_B \\
        & =\util^*(E)-\left(1-\frac{2\sqrt{n}-1}{n}\right)\mathrm{OPT}_B.
    \end{align*}
    Since $B$ consists of the $n$ smallest values among $\{a_t\}_{t\in[\ell]}$, its average is at most the overall average.
    Thus,
    \[
    \mathrm{OPT}_B
    \le \frac{n}{\ell}\sum_{t=1}^{\ell} a_t
    = \frac{n}{\ell}\cdot \util^*(E).
    \]
    Indeed, sorting $a_1,\dots,a_\ell$ in nondecreasing order, $\mathrm{OPT}_B$ is the sum of the first $n$, whose average cannot exceed the overall average.
    Substituting gives
    \begin{align*}
        \util(\mathbf{o}) & \ge \left[1-\left(1-\frac{2\sqrt{n}-1}{n}\right)\frac{n}{\ell}\right]\util^*(E) \\
        & = \frac{\ell-n+2\sqrt{n}-1}{\ell}\cdot \util^*(E).
    \end{align*}
    Finally, by construction, every voter is satisfied at least once in $\mathbf{o}$ (indeed, already within $B$), i.e.,
    $\sat_i(\mathbf{o})\ge 1$ for all $i\in N$.
    Therefore, for any nonempty voter set $S\subseteq N$ we have $\sat_S(\mathbf{o})\ge \sat_i(\mathbf{o})\ge 1$ for any $i\in S$.
    Since the right-hand side of the JR constraint is $\min\{1,\lfloor t|S|/n\rfloor\}\in\{0,1\}$, this implies that all JR constraints
    are satisfied; thus, $\mathbf{o}$ provides JR. Consequently $\util_{\JR}(E)\ge \util(\mathbf{o})$, and thus
    \[
    \frac{\util^*(E)}{\util_{\JR}(E)}
    \le \frac{\ell}{\ell-n+2\sqrt{n}-1}.
    \]
    Taking the supremum over complete instances concludes the proof.

\subsection{Proof of Proposition~\ref{prop:jr-tight}}
    Fix $n,\ell$ with $\ell\ge n$, and let $k:=\lceil \sqrt n\rceil$.
    We construct a complete temporal election
    $E=(P,N,\ell,(s_i)_{i\in N})$ as follows.
    
    Partition the voters into two disjoint sets $N=C \dot\cup R$ with
    $|C|=k$ (the \emph{core} voters) and $|R|=n-k$ (the \emph{private} voters).
    Let $P := \{z\} \cup \{p_i : i\in R\}$, where all candidates are distinct.
    
    For every round $t\in[\ell]$, define approval sets by
    \begin{equation*}
        s_{c,t}:=\{z\}\quad\text{for all }c\in C, \text{ and } s_{i,t}:=\{p_i\}\quad\text{for all }i\in R.
    \end{equation*}
    Then $E\in \mathcal{E}_{\ge 1}$ (every voter approves at least one candidate in every round), and preferences are static.
    
    Now, in every round $t$, candidate $z$ is approved by exactly the $k$ voters in $C$, while each $p_i$ is approved by exactly one voter.
    Thus, the maximum possible per-round welfare is $k$, attained by choosing $z$.
    Therefore the unconstrained utilitarian optimum is achieved by
    $\mathbf{o}^*:=(z,z,\dots,z)$ and
    \[
    \util^*(E)=k\ell.
    \]
    
    Let $\mathbf{o}=(o_1,\dots,o_\ell)\in P^\ell$ be any outcome satisfying JR.
    We first show that every voter must be satisfied at least once.
    Fix any voter $i\in N$ and consider the singleton group $S:=\{i\}$.
    Since $E$ is complete, $S$ agrees in every round, so it agrees in a size-$\ell$ subset of rounds.
    Applying Definition~\ref{defn:jr_pjr_ejr} with $t=\ell$ gives the JR requirement
    \begin{equation*}
        \sat_S(\mathbf{o}) \geq \min\{1,\lfloor \ell\cdot |S|/n\rfloor\}
        =\min\{1,\lfloor \ell/n\rfloor\}=1,
    \end{equation*}
    where the last equality uses $\ell\ge n$.
    Since $S$ is a singleton, $\sat_S(\mathbf{o})=\sat_i(\mathbf{o})$, thus, $\sat_i(\mathbf{o})\ge 1$ for every $i\in N$.
    
    In particular, every private voter $i\in R$ must be satisfied in some round.
    But $i$ approves only $p_i$ in every round, so there exists a round $t(i)\in[\ell]$ with $o_{t(i)}=p_i$.
    If $i\neq j$ are two private voters, then $\{p_i\}\cap\{p_j\}=\varnothing$, so they cannot share the same witness round:
    $t(i)\neq t(j)$.
    Thus $\mathbf{o}$ selects a private candidate in at least $|R|=n-k$ distinct rounds.
    Equivalently, if $r:=\bigl|\{t\in[\ell]: o_t=z\}\bigr|$ is the number of rounds selecting $z$, then $\ell-r\ge n-k$, i.e., $r\le \ell-n+k$.
    
    A round with $o_t=z$ contributes exactly $k$ to welfare, while a round with $o_t\neq z$ contributes exactly $1$.
    Thus,
    \begin{align*}
        \util(\mathbf{o})=rk+(\ell-r)\cdot 1 & =\ell+r(k-1) \\
        & \le \ell+(\ell-n+k)(k-1) \\
        & = k\ell-(k-1)(n-k).
    \end{align*}
    This upper bound is achieved by selecting each private candidate $p_i$ once (for $i\in R$) and selecting $z$ in all remaining $\ell-(n-k)$ rounds; this outcome satisfies every voter at least once and therefore satisfies JR.
    Consequently,
    \[
    \util_{\JR}(E)=k\ell-(k-1)(n-k).
    \]
    For this instance,
    \[
    \frac{\util^*(E)}{\util_{\JR}(E)}
    =\frac{k\ell}{k\ell-(k-1)(n-k)}
    =\frac{\ell}{\ell-n+\frac{n}{k}+k-1}.
    \]
    Now set $s:=\sqrt{n}$ and write $k=s+\delta$ for some $\delta\in[0,1)$ (since $k=\lceil s\rceil$).
    Then
    \begin{align*}
        \frac{n}{k}+k-1 & = \frac{s^2}{s+\delta}+s+\delta-1 \\
        & =2s-1+\frac{\delta^2}{s+\delta}
    =2\sqrt{n}-1+\frac{\delta^2}{k}
    \le 2\sqrt{n}-\frac12,
    \end{align*}
    where the last inequality uses $\delta^2<1$ and $k\ge 2$ for $n\ge 2$ (and the case $n=1$ can be checked directly).
    Therefore
    \[
    \frac{\util^*(E)}{\util_{\JR}(E)}
    \ \ge\ \frac{\ell}{\ell-n+2\sqrt{n}-\frac12}
    \ =\ \frac{\ell}{\ell-n+2\sqrt{n}-\mathcal{O}(1)}.
    \]
    Since $\mathcal{P}_\util(\JR)$ is defined as a supremum over instances, this establishes the claimed lower bound.

\subsection{Proof of Theorem~\ref{thm:long-horizon-sep}}
    Fix an integer $a\ge 2$.
    Note that Theorem~\ref{thm:jr-ub} shows that for every election with $n$ voters and $\ell\ge n$ rounds,
    \[
    \frac{\util^*(E)}{\util_{\JR}(E)} \le \frac{\ell}{\ell-n+2\sqrt{n}-1}.
    \]
    Substituting $\ell=an$ gives
    \begin{align*}
        \mathcal{P}_\util(\JR) & \le \frac{an}{an-n+2\sqrt{n}-1} \\
        & =\frac{a}{(a-1)+\frac{2}{\sqrt{n}}-\frac{1}{n}} \\
        & = \frac{a}{a-1}+o(1) \text{ as } n\to\infty.
    \end{align*}
    
    For each $n$, let $k:=\lceil \sqrt{n}\rceil$ and set $\ell:=an$.
    Construct a temporal election
    $E=(P,N,\ell,(s_i)_{i\in N})$ where agents have static preferences, as follows.
    Partition the voters as $N=C\dot\cup R$ with $|C|=k$ (core voters) and $|R|=n-k$ (private voters).
    Let
    \begin{equation*}
        P := \{z\} \cup \{p_i : i\in R\},
    \end{equation*}
    where all candidates are distinct.
    For every round $r\in[\ell]$, define approval sets by
    \[
    s_{c,r}:=\{z\}\ \text{ for all }c\in C \text{ and }
    s_{i,r}:=\{p_i\}\ \text{ for all }i\in R.
    \]
    Then, $E \in \mathcal{E}_{\geq 1}$ (every $s_{i,r}\neq\varnothing$) and in every round $z$ is approved by exactly
    $k$ voters, while every $p_i$ is approved by exactly one voter.
    
    In each round, selecting $z$ gives us welfare $k$, while selecting any $p_i$ gives us welfare $1$.
    Thus, the unconstrained optimum is to pick $z$ in every round, giving
    \[
    \util^*(E)=k\ell = akn.
    \]
    Now fix $\Phi\in\{\PJR,\EJR,\EJRplus\}$ and let $\mathbf{o}=(o_1,\dots,o_\ell)$ be any $\Phi$-feasible outcome.

    We first prove that every voter must be satisfied at least $a$ times.
    Fix any voter $i\in N$ and consider the singleton group $S:=\{i\}$.
    Since $E \in \mathcal{E}_{\geq 1}$, $S$ agrees in every round, thus, in a size-$\ell$ subset of rounds.
    Moreover, $\lfloor \ell/n\rfloor=\lfloor an/n\rfloor=a$.
    
    \begin{itemize}
        \item If $\Phi\in\{\PJR,\EJR\}$, apply Definition~\ref{defn:jr_pjr_ejr} with $t=\ell$ to obtain
        $\sat_S(\mathbf{o})=\sat_i(\mathbf{o})\ge \lfloor \ell/n\rfloor=a$.
        \item If $\Phi=\EJRplus$, take $\sigma:=1$ and $\tau:=\ell$.
        Since $s_{i,r}\neq\varnothing$ for all $r$, the singleton $S$ is $(1,\ell)$-cohesive (witnessed by
        choosing in each round an arbitrary candidate in $s_{i,r}$).
        Applying Definition~\ref{defn:ejrplus} to this $S$ gives, for every round $r$, that either
        (i) $\sat_i(\mathbf{o})\geq \lfloor \tau\sigma/n\rfloor=\lfloor \ell/n\rfloor=a$ or (ii) $o_r\in s_{i,r}$.
        Note that (i) depends only on the full outcome (not on $r$), so if it is false then (ii) must hold for every $r$.
        If (i) holds we are done; otherwise (ii) holds for all $r$, which implies $\sat_i(\mathbf{o})=\ell\ge a$.
        Thus in all cases $\sat_i(\mathbf{o})\ge a$.
    \end{itemize}

    Next, we show that at least $a(n-k)$ rounds must select private candidates.
    Each private voter $i\in R$ approves only $p_i$, so $\sat_i(\mathbf{o})$ equals the number of rounds with $o_r=p_i$.
    Then, $\sat_i(\mathbf{o})\ge a$, so $p_i$ must be selected at least $a$ times.
    Since each round can select only one candidate, the total number of rounds selecting private candidates is at least $a|R|=a(n-k)$.
    Let $r_z := \bigl|\{r\in[\ell]: o_r=z\}\bigr|$  denote the number of rounds in which $z$ is selected. Then
    \[
    r_z \le \ell - a(n-k) = an - a(n-k) = ak.
    \]
    Now, every round with $o_r=z$ contributes exactly $k$ to welfare, while every other round contributes exactly $1$.
    Therefore,
    \begin{align*}
        \util(\mathbf{o}) 
        = r_zk+(\ell-r_z) & = \ell+r_z(k-1) \\
        & \le \ell+ak(k-1) \\
        & =an+ak^2-ak \\
        & =a(n-k+k^2).
    \end{align*}
    This bound is achievable by selecting each $p_i$ (for $i\in R$) exactly $a$ times and selecting $z$ in the
    remaining $ak$ rounds; this outcome satisfies $\sat_i(\mathbf{o})\ge a$ for all $i\in N$, and (as noted above) the only nontrivial cohesive groups are subsets of $C$ and singleton private voters, all of whose demands are met.
    More concretely, under static preferences, a voter set $S$ agrees (i.e., has nonempty intersection) iff either $S\subseteq C$ (intersection $\{z\}$) or $S=\{i\}$ for some $i\in R$ (intersection $\{p_i\}$); any set containing both a core and a private voter (or two distinct private voters) has empty intersection and imposes no JR/PJR/EJR constraint.
    If $S\subseteq C$ with $|S|=s$, then $\lfloor \ell s/n\rfloor=\lfloor (an)s/n\rfloor = as$, and since $z$ is chosen $ak$ times, we have $\sat_S(\mathbf{o})=ak\ge as$ (PJR) and each $c\in C$ has $\sat_c(\mathbf{o})=ak\ge as$ (EJR).
    If $S=\{i\}\subseteq R$, then $\sat_i(\mathbf{o})=a$, meeting the singleton demand $a$.
    For EJR+, the only sets $S$ that create constraints are again those with nonempty intersection (i.e., $S\subseteq C$ or $S=\{i\}\subseteq R)$; in both cases every voter in $S$ has satisfaction at least $a$ and in fact core voters have satisfaction $ak$, which is at least $\lfloor \tau\sigma/n\rfloor$ for all admissible $\tau\le \ell and \sigma\le |S|$, so (i) holds.

    Consequently,
    \[
    \util^\Phi(E)= a(n-k+k^2).
    \]
    We obtain
    \[
    \frac{\util^*(E)}{\util^\Phi(E)}
    = \frac{akn}{a(n-k+k^2)}
    = \frac{kn}{n-k+k^2}.
    \]
    With $k=\lceil \sqrt{n}\rceil$ we have $k\ge \sqrt{n}$ and $k\le \sqrt{n}+1$, thus,
    $n-k+k^2 \le n+k^2 \le n+(\sqrt{n}+1)^2 = 2n + 2\sqrt{n} + 1 \le 5n$ for all $n\ge 1$.
    Therefore
    \[
    \frac{\util^*(E)}{\util^\Phi(E)} \ge \frac{\sqrt{n}\cdot n}{5n} = \frac{1}{5}\sqrt{n},
    \]
    so in particular $\mathcal{P}_\util(\Phi)=\Omega(\sqrt{n})$ for each $\Phi\in\{\PJR,\EJR,\EJRplus\}$ in the case where $\ell=an$.

\section{Omitted Proofs in Section~\ref{sec:computational}}

\subsection{Proof of Lemma~\ref{lem:alpha01-collapse}}
    Fix any integer $t>0$ and any nonempty voter subset $S\subseteq N$ that agrees in a size-$t$ subset of rounds. Set $d := \left\lfloor \frac{t\cdot |S|}{n}\right\rfloor$.
    By the assumption of the lemma, $d\in\{0,1\}$.
    
    If $d=0$, then JR/PJR require $\sat_S(\mathbf{o})\geq 0$, and EJR requires $\exists i\in S$ with $\sat_i(\mathbf{o})\ge 0$, which holds since $S\neq\varnothing$.
    
    Assume $d=1$. Then the JR and PJR requirements for $(S,t)$ both reduce to $\sat_S(\mathbf{o})\ge 1$. Moreover,
    \begin{align*}
        & \sat_S(\mathbf{o})\ge 1 \\
        & \iff \exists r\in[\ell]\ \exists i\in S:\ o_r\in s_{i,r} \\
        & \iff  \exists i\in S:\ \sat_i(\mathbf{o})\ge 1,
    \end{align*}
    so the EJR requirement for $(S,t)$ is equivalent to the same condition. Thus, for every such pair $(S,t)$, JR, PJR, and EJR impose identical constraints on $\mathbf{o}$, and thus, the three axioms are equivalent.

\subsection{Proof of Lemma~\ref{lem:static_ejr_ejr+}}
    We first show that every EJR+ outcome is EJR (this direction does not rely on static preferences).
    Let $\mathbf{o}$ satisfy EJR+ and fix any integer $t>0$ and any nonempty voter subset $S\subseteq N$ that agrees in a size-$t$ subset of rounds. Let $R\subseteq[\ell]$ with $|R|=t$ witness this agreement, i.e., $\bigcap_{i\in S}s_{i,r}\neq\varnothing$ for all $r\in R$. For each $r\in R$ choose any $p_r\in\bigcap_{i\in S}s_{i,r}$ and define $\tilde{o}_R := (p_r)_{r\in R}$. Then $S$ is $(|S|,t)$-cohesive. Applying EJR+ to the parameters $\sigma:=|S|$, $\tau:=t$, the set $S$, and any round $r\in R$ gives us either:
    \begin{enumerate}
    \item[(i)] there exists $i\in S$ with $\sat_i(\mathbf{o})\ge \left\lfloor \frac{t|S|}{n}\right\rfloor$, which is exactly the EJR requirement for $(S,t)$; or
    \item[(ii)] $o_r\in\bigcap_{i\in S}s_{i,r}$.
    \end{enumerate}
    If (i) holds we are done. Otherwise (i) is false, so by EJR+ we must have (ii) for every $r\in R$. Thus, every voter $i\in S$ approves the chosen candidate in each round of $R$, and therefore $\sat_i(\mathbf{o})\ge t$ for all $i\in S$. 
    Since $|S| \leq n$, we have that $\left\lfloor \frac{t|S|}{n}\right\rfloor\le t$, and the EJR requirement again holds. 
    Thus $\mathbf{o}$ satisfies EJR.
    
    Now assume preferences are static and that $\mathbf{o}$ satisfies EJR. Let $\sigma\in[n]$, $\tau\in[\ell]$, let $S$ be any $(\sigma,\tau)$-cohesive subset, and let $r\in[\ell]$ be any round with $\bigcap_{i\in S}s_{i,r}\neq\varnothing$. Under static preferences, $\bigcap_{i\in S}s_{i,r}\neq\varnothing$ implies $\bigcap_{i\in S}s_{i,r'}\neq\varnothing$ for every $r'\in[\ell]$, so $S$ agrees in a size-$\ell$ subset of rounds. Applying EJR to $(S,\ell)$ gives us some $i\in S$ with
    \[
    \sat_i(\mathbf{o})\ge \left\lfloor \frac{\ell|S|}{n}\right\rfloor.
    \]
    Since $S$ is $(\sigma,\tau)$-cohesive we necessarily have $|S|\ge \sigma$, and we also have $\ell\ge \tau$; thus,
    \[
    \left\lfloor \frac{\ell|S|}{n}\right\rfloor \ge \left\lfloor \frac{\tau\sigma}{n}\right\rfloor.
    \]
    Consequently the same voter $i$ satisfies $\sat_i(\mathbf{o})\ge \left\lfloor \frac{\tau\sigma}{n}\right\rfloor$, so the EJR+ condition holds via (i). 
    As $\sigma,\tau,S,$ and $r$ were arbitrary, $\mathbf{o}$ satisfies EJR+.

\subsection{Proof of Theorem~\ref{thm:intractable_NPc}}
    Fix $\Phi \in \{\JR,\PJR,\EJR,\EJRplus\}$.
    We prove NP-hardness via a reduction from \textsc{Exact Cover by 3-Sets} (X3C), which is known to be NP-complete.
    An instance of X3C consists of a universe $U=\{1,\dots,u\}$ with $u=3q$ and a family $\mathcal{S}=\{S_1,\dots,S_m\}$ of 3-element subsets of $U$; it is a yes-instance if there exist $q$ sets whose union is exactly $U$, and a no-instance otherwise.
    Without loss of generality, assume every element of $U$ appears in at least one set (otherwise the instance is trivially a no-instance and we can map it to any fixed no-instance of $\Phi$-\textsc{UTIL}).
    We may also assume $q\ge 1$; constant-size instances can be handled separately.
    
    Given $(U,\mathcal{S})$, we build a temporal election $E=(P,N,\ell,(s_i)_{i\in N})$
    with static preferences as follows.
    \begin{itemize}
      \item \emph{Voters:} for every element $e\in U$, create three voters
      $x_e,y_e^1,y_e^2$. Additionally, create four dummy voters $d_1,d_2,d_3,d_4$.
      Thus $n:=3u+4=9q+4$.
      \item \emph{Rounds:} set $\ell:=3q+2$.
      \item \emph{Candidates:} for every set $S_j\in \mathcal{S}$, create a
      candidate $c_j$; for every element $e\in U$, create a candidate $p_e$; and
      create one special candidate $z$.
      \item \emph{Static preferences:} for every $e\in U$, define
      $A_{y_e^1}=A_{y_e^2}=\{p_e\}$ and
      $A_{x_e}=\{p_e\}\cup \{c_j : e\in S_j\}$.
      For each dummy voter $d_r$, let $A_{d_r}=\{z\}$.
      Set $s_{i,t}=A_i$ for all voters $i$ and all rounds $t$ (hence preferences are static).
    \end{itemize}
    By construction, each set-candidate $c_j$ is approved by exactly the three
    voters $\{x_e : e\in S_j\}$, each element-candidate $p_e$ is approved by exactly
    $\{x_e,y_e^1,y_e^2\}$, and $z$ is approved by exactly the four dummy voters.
    Thus, each candidate is approved by at most $4$ voters.
    
    Finally, set the welfare threshold
    \[
    B := 4\ell - q.
    \]
    For any outcome $\mathbf{o}=(o_1,\dots,o_\ell)$, define
    \[
    w(\mathbf{o}):=\bigl|\{t\in[\ell] : o_t \neq z\}\bigr|.
    \]
    Every round with $o_t=z$ contributes $4$ to the round welfare, while every round
    with $o_t\neq z$ contributes $3$. Therefore
    \[
    \util(\mathbf{o})=4(\ell-w(\mathbf{o}))+3w(\mathbf{o})=4\ell-w(\mathbf{o}),
    \]
    and thus
    \[
    \util(\mathbf{o})\ge B \iff w(\mathbf{o})\le q.
    \]

    Let $S\subseteq N$ be any voter subset and let $t\in[\ell]$.
    Because preferences are static, if $S$ agrees in a size-$t$ subset of rounds,
    then $\bigcap_{i\in S} A_i \neq \varnothing$, so $S$ is contained in the approvers
    of some single candidate and hence $|S|\le 4$.
    Consequently,
    \[
    \Bigl\lfloor \frac{t\cdot |S|}{n}\Bigr\rfloor \;\le\;
    \Bigl\lfloor \frac{\ell\cdot 4}{n}\Bigr\rfloor
    \;=\;
    \Bigl\lfloor \frac{4(3q+2)}{9q+4}\Bigr\rfloor
    \;=\; 1,
    \]
    where the last equality uses $q\ge 1$, so $4(3q+2)<2(9q+4)$ and $4(3q+2)\ge 9q+4$.
    Thus, in this instance, every threshold of the form
    $\lfloor t|S|/n\rfloor$ is in $\{0,1\}$.
    For EJR+, note that constraints are only for rounds $r$ where $\bigcap_{i \in S} s_{i,r} \neq \varnothing$.
    In our instance, any such $S$ must be contained in the approvers of a single candidate, and thus $|S| \leq 4$, giving us $\sigma \leq 4$.
    Thus, we get that
    \[
    \Bigl\lfloor \frac{\tau\cdot \sigma}{n}\Bigr\rfloor \leq
    \Bigl\lfloor \frac{4\ell}{n}\Bigr\rfloor = 1.
    \]
    
    We first prove yes-instance $\implies$ feasible high welfare outcome.
    Assume the X3C instance is a yes-instance, and fix $q$ sets
    $S_{j_1},\dots,S_{j_q}$ whose union is exactly $U$.
    Define an outcome $\mathbf{o}$ by selecting $c_{j_k}$ in round $k$ for $k=1,\dots,q$, and
    selecting $z$ in the remaining $\ell-q$ rounds.
    Then $w(\mathbf{o})=q$, so $\util(\mathbf{o})=4\ell-q=B$.
    
    We verify that $\mathbf{o}$ satisfies $\Phi$.
    First note that every dummy voter $d_r$ has $\sat_{d_r}(\mathbf{o})\ge 1$ because
    $\mathbf{o}$ selects $z$ in $\ell-q=2q+2\ge 1$ rounds.
    Next, for each element $e\in U$, since $\{S_{j_1},\dots,S_{j_q}\}$ is an exact
    cover, there is a unique $k$ such that $e\in S_{j_k}$, and then $x_e$ approves
    $c_{j_k}$; hence $\sat_{x_e}(\mathbf{o})\ge 1$ for all $e\in U$.
    
    Now consider any constraint relevant to $\Phi$.
    \begin{itemize}
      \item If $\Phi\in\{\JR,\PJR,\EJR\}$, fix any $t>0$ and
      any $S\subseteq N$ that agrees in a size-$t$ subset of rounds. If
      $\lfloor t \cdot |S|/n\rfloor=0$ the constraint is trivial. Otherwise
      $\lfloor t \cdot |S|/n\rfloor=1$, and since $t\le \ell=3q+2$ and $n=9q+4>2\ell$, if $|S|\leq 2$, then $t \cdot|S|\le 2\ell<n$; thus we must have that $|S|\ge 3$.
      Then, (as argued above) $S$ is a subset of the approvers of a single candidate, and thus it is contained in one of the following approver sets: $\{d_1,d_2,d_3,d_4\}$, $\{x_e,y_e^1,y_e^2\}$ for some
      $e$, or $\{x_{e_1},x_{e_2},x_{e_3}\}$ for some 3-set $\{e_1,e_2,e_3\}$.
      In all cases, $S$ contains at least one voter with satisfaction at least $1$
      (a dummy voter, or some $x_e$), so $\JR$, $\PJR$, and
      $\EJR$ are satisfied for this $(S,t)$.
      \item If $\Phi=\EJRplus$, fix any $\sigma\in[n]$, $\tau\in[\ell]$, any
      $(\sigma,\tau)$-cohesive set $S$, and any round $r$ with
      $\bigcap_{i\in S} s_{i,r}\neq\varnothing$. If $\lfloor \tau\sigma/n\rfloor=0$,
      condition (i) holds trivially. Otherwise $\lfloor \tau\sigma/n\rfloor=1$,
      which implies $\sigma\ge 3$ and thus $|S|\ge 3$. The nonempty intersection
      implies again that $S$ is contained in the approvers of a single candidate, so
      by the same case analysis as above, $S$ contains a voter with satisfaction at
      least $1$. Hence $\sat_i(\mathbf{o})\ge \lfloor\tau\sigma/n\rfloor$ for some
      $i\in S$, i.e., $\EJRplus$ condition (i) holds.
    \end{itemize}
    Therefore the constructed $\Phi$-\textsc{UTIL} instance is a yes-instance.
    
    Next we prove no-instance $\implies$ no feasible high welfare outcome.
    Assume the X3C instance is a no-instance. Suppose for contradiction that there
    exists an outcome $\mathbf{o}$ satisfying $\Phi$ with $\util(\mathbf{o})\ge B$.
    Then $w(\mathbf{o})\leq q$.
    
    Fix any element $e\in U$ and consider the element group
    \[
    G_e := \{x_e,y_e^1,y_e^2\}.
    \]
    Since all three voters approve $p_e$, we have $\bigcap_{i\in G_e} A_i=\{p_e\}\neq\varnothing$,
    so $G_e$ agrees in a size-$\ell$ subset of rounds. Moreover,
    \[
    \Bigl\lfloor \frac{\ell\cdot |G_e|}{n}\Bigr\rfloor
    =
    \Bigl\lfloor \frac{3\ell}{n}\Bigr\rfloor
    =
    \Bigl\lfloor \frac{3(3q+2)}{9q+4}\Bigr\rfloor
    =
    \Bigl\lfloor \frac{9q+6}{9q+4}\Bigr\rfloor
    =1.
    \]
    Thus, by Definition~\ref{defn:jr_pjr_ejr}, if $\Phi\in\{\JR,\PJR\}$ we obtain
    $\sat_{G_e}(\mathbf{o})\ge 1$, and if $\Phi=\EJR$ we obtain the existence of some $i\in G_e$ with $\sat_i(\mathbf{o})\ge 1$.
    If $\Phi=\EJRplus$, note that $G_e$ is $(3,\ell)$-cohesive (select $p_e$
    in all $\ell$ rounds), and for every round $r$ we have
    $\bigcap_{i\in G_e} s_{i,r}=\{p_e\}\neq\varnothing$, so applying Definition~\ref{defn:ejrplus}
    implies either (i) some $i\in G_e$ has $\sat_i(\mathbf{o})\ge \lfloor 3\ell/n\rfloor=1$
    or (ii) $o_r=p_e$; in either case, at least one voter in $G_e$ is satisfied in
    at least one round. Therefore, for every $e\in U$ there exists some round $t$
    such that $o_t$ is approved by at least one voter in $G_e$.
    
    But the only candidates approved by any voter in $G_e$ are $p_e$ (approved by all
    three) and the set-candidates $c_j$ with $e\in S_j$ (approved by $x_e$).
    Hence, for each $e\in U$, the outcome must select either $p_e$ in some round or
    some $c_j$ with $e\in S_j$ in some round.
    
    Let $w_{\mathrm{set}} := \bigl|\{t\in[\ell] : o_t=c_j \text{ for some } j\}\bigr|$ and $w_{\mathrm{elm}} := \bigl|\{t\in[\ell] : o_t=p_e \text{ for some } e\}\bigr|$.
    Then $w_{\mathrm{set}}+w_{\mathrm{elm}}=w(\mathbf{o})\le q$.
    A set-round (selecting some $c_j$) can certify the above condition for at most
    the three elements in $S_j$, whereas an element-round (selecting some $p_e$)
    can certify it for exactly one element.
    Thus the number of \emph{distinct} covered elements is at most
    $3w_{\mathrm{set}}+w_{\mathrm{elm}}$, and therefore
    \[
    |U|=3q \;\le\; 3w_{\mathrm{set}}+w_{\mathrm{elm}}
    \le 3(w_{\mathrm{set}}+w_{\mathrm{elm}}) \;\le\; 3q.
    \]
    All inequalities must be tight. In particular, $w_{\mathrm{set}}+w_{\mathrm{elm}}=q$
    and $w_{\mathrm{elm}}=0$, so $\mathbf{o}$ uses exactly $q$ set-rounds and no element-rounds.
    Tightness also forces these $q$ set-rounds to cover $3q$ distinct elements, hence
    their underlying 3-sets are pairwise disjoint and have union $U$, resulting in an exact cover of size $q$. This contradicts that the X3C instance is a no-instance.
    Therefore, the constructed $\Phi$-\textsc{UTIL} instance is a no-instance, and
    $\Phi$-\textsc{UTIL} is NP-hard.
    
    Finally, we show membership in NP.
    Given a certificate outcome $\mathbf{o} \in P^\ell$, we can compute $\util(\mathbf{o})$ and
    all $\sat_i(\mathbf{o})$ in $\mathcal{O}(n\ell)$ time.
    To verify $\Phi$ under static preferences and the promise that each candidate has at most $4$ approvers, it suffices to enumerate all voter groups $S\subseteq N$ with $\bigcap_{i\in S} A_i\neq \varnothing$:
    for any such $S$ pick any $p\in \bigcap_{i\in S} A_i$; then
    $S\subseteq \{i\in N : p\in A_i\}$, so $|S|\le 4$ and thus there are at most
    $2^4$ such groups per candidate.
    For $\JR/\PJR/\EJR$, under static preferences, any group either agrees in all rounds or in none, so it suffices to check only $t=\ell$ because $t\mapsto \lfloor t \cdot |S|/n\rfloor$ is nondecreasing and the left-hand side does not depend on $t$.
    For $\Phi = \EJRplus$, membership in NP follows from Lemma~\ref{lem:static_ejr_ejr+}, and we can verify EJR in polynomial time.
    Thus, $\Phi$-\textsc{UTIL} is in NP.
    
    Therefore, $\Phi$-\textsc{UTIL} is NP-complete, even when voters have static preferences and each candidate is approved by at most $4$ voters.

\subsection{Proof of Theorem~\ref{thm:intractable_APXh}}
    We give a PTAS-reduction from \textsc{Min-Vertex-Cover} on cubic graphs, which is known to be APX-complete \cite{AlimontiKann}.
    
    Fix $\Phi\in\{\JR,\PJR,\EJR,\EJRplus\}$. If $\Phi=\EJRplus$, then by Lemma~\ref{lem:static_ejr_ejr+}, it suffices to prove the claim for $\EJR$. Hence, throughout the proof we assume $\Phi\in\{\JR,\PJR,\EJR\}$.
    
    Let $G=(V,E)$ be a cubic graph with $|V|=m$ and $|E|=3m/2$.
    We construct a temporal election $E_G=(P,N,\ell,(A_i)_{i\in N})$
    with static preferences as follows.
    \begin{itemize}
        \item \emph{Candidates. }
        \begin{equation*}
            P := \{z\}\ \cup\ \{c_v : v\in V\}\ \cup\ \{p_e : e\in E\}.
        \end{equation*}
        \item \emph{Voters and approvals.} Define three types of voters.
        \begin{itemize}
            \item \emph{Baseline voters:} voters $b_1,\dots,b_8$ with $A_{b_j}=\{z\}$.
            \item \emph{Vertex supporters:} for each $v\in V$, voters $s_{v,1},s_{v,2}$ with
            $A_{s_{v,k}}=\{c_v\}$ for $k\in\{1,2\}$.
            \item \emph{Edge gadget:} for each edge $e=\{u,v\}\in E$, voters $a_{e,1},a_{e,2},a_{e,3},a_{e,4}$ with $A_{a_{e,1}}=A_{a_{e,2}}=A_{a_{e,3}}=\{p_e\}$ and $A_{a_{e,4}}=\{p_e,c_u,c_v\}$.
        \end{itemize}
        Thus every voter approves at most $3$ candidates.
    \end{itemize}
    Then, the number of voters is
    \begin{equation*}
        n:=8+2|V|+4|E|=8+2m+4\cdot\frac{3m}{2}=8m+8=8(m+1).
    \end{equation*}
    Set the number of rounds
    \begin{equation*}
        \ell := \frac{n}{4}=2(m+1)=2m+2.
    \end{equation*}
    For each candidate $x\in P$, define its approval weight
    \[
    w(x):=\bigl|\{i\in N : x\in A_i\}\bigr|.
    \]
    By construction, $w(z)=8$, $w(p_e)=4$ for all $e \in E$, and $w(c_v)=2+\deg_G(v)=5$ for all $v \in V$, since $G$ is cubic, where $\deg_G(\cdot)$ is a function that returns the degree of the vertex in $G$. 
    Under static preferences, for any outcome $\mathbf{o}=(o_1,\dots,o_\ell)\in P^\ell$, we have
    \begin{equation}\label{eq:static-welfare}
        \util(\mathbf{o})=\sum_{r=1}^{\ell} w(o_r).
    \end{equation}    

    Now, fix an edge $e=\{u,v\}\in E$ and define the voter set
    \begin{equation*}
        S_e := \{a_{e,1},a_{e,2},a_{e,3},a_{e,4}\}.
    \end{equation*}
    All voters in $S_e$ approve $p_e$, thus (under static preferences) $S_e$ agrees in all $\ell$ rounds, and
    \[
    \left\lfloor \frac{\ell\cdot |S_e|}{n}\right\rfloor
    =
    \left\lfloor \frac{(n/4)\cdot 4}{n}\right\rfloor
    =1.
    \]
    Therefore, for every $\Phi\in\{\JR,\PJR,\EJR\}$, providing $\Phi$ (i.e., $\Phi$-feasible) implies that $S_e$ is represented at least once, i.e.,
    \begin{equation}\label{eq:edge-rep}
    \exists\, t\in[\ell]\ \text{such that}\ o_t \in \bigcup_{i\in S_e} A_i = \{p_e,c_u,c_v\}.
    \end{equation}
    
    Next, we show that we may restrict our attention to a simple class of outcomes.
    
    \medskip
    \noindent\emph{Claim 1.} There exists an optimal $\Phi$-feasible outcome that never selects any $p_e$.
    
    \smallskip
    \noindent\emph{Proof.}
    Let $\mathbf{o}$ be any outcome that provides $\Phi$ (i.e., $\Phi$-feasible), and suppose $o_t=p_e$ for some $e=\{u,v\}$.
    Construct $\mathbf{o}'$ by replacing $o_t$ with $c_u$ (either endpoint works).
    Condition~\eqref{eq:edge-rep} for edge $e$ remains satisfied because $c_u\in\{p_e,c_u,c_v\}$.
    No other representation constraint can be violated by this replacement: the only voters who approve $p_e$ are exactly the members of $S_e$, and after the replacement the voter $a_{e,4}\in S_e$ is still satisfied in round $t$ (since $c_u\in A_{a_{e,4}}$).
    Finally, by~\eqref{eq:static-welfare} the welfare strictly increases because $w(c_u)=5>w(p_e)=4$.
    Moreover, any cohesive group whose common approved candidate is $p_e$ is contained in $S_e$; and since $\ell=n/4$, every proper subset $S\subsetneq S_e$ has demand $\lfloor \ell|S|/n\rfloor=\lfloor |S|/4\rfloor=0$, so the only potentially binding constraint involving $p_e$ is for $S_e$ itself.
    Repeating this replacement exhaustively gives us an optimal $\Phi$-feasible outcome with no $p_e$. \qed
    
    \medskip
    \noindent\emph{Claim 2.} There exists an optimal $\Phi$-feasible outcome in which each $c_v$ appears in at most one round.
    
    \smallskip
    \noindent\emph{Proof.}
    Take an optimal $\Phi$-feasible outcome with no $p_e$ (Claim~1).
    If some $c_v$ appears in at least two rounds, keep one occurrence and replace each additional occurrence by $z$.
    This strictly increases welfare since $w(z)=8>w(c_v)=5$.
    We argue that $\Phi$-feasibility is preserved.
    
    First, for every edge $e$ incident to $v$, condition~\eqref{eq:edge-rep} only requires \emph{existence} of a witnessing round, and the remaining single occurrence of $c_v$ still witnesses~\eqref{eq:edge-rep} for all edges incident to $v$.
    
    Second, any cohesive group $S$ with $c_v \in \bigcap_{i\in S}A_i$ must satisfy $S\subseteq \{i\in N: c_v\in A_i\}$, which has size $w(c_v)=5$.
    Since $\ell=n/4$, we have
    \[
    \left\lfloor \frac{\ell\cdot |S|}{n}\right\rfloor
    \le
    \left\lfloor \frac{\ell\cdot 5}{n}\right\rfloor
    =
    \left\lfloor \frac{5}{4}\right\rfloor
    =1.
    \]
    Thus any $\Phi$-requirement involving such an $S$ can be satisfied already with a single round in which some voter in $S$ approves the chosen candidate; keeping one occurrence of $c_v$ suffices, and replacing further occurrences by $z$ cannot destroy feasibility.
    \qed
    
    \medskip
    By Claims~1--2, there exists an optimal $\Phi$-feasible outcome of the following canonical form:
    choose a set $U\subseteq V$, select each $c_v$ for $v\in U$ in exactly one round, and select $z$ in the remaining $\ell-|U|$ rounds.

    Now, we prove that there exists a vertex cover of $G$ if and only if there exists a feasible canonical outcome in our constructed instance.

    For the `if' direction, let $\mathbf{o}$ be canonical with corresponding $U\subseteq V$.
    Since $\mathbf{o}$ never selects any $p_e$, condition~\eqref{eq:edge-rep} implies that for each edge $e=\{u,v\}$ at least one of $c_u,c_v$ is selected, i.e., at least one of $u,v$ lies in $U$.
    Thus $U$ is a vertex cover of $G$.
    
    For the `only if' direction, let $U$ be a vertex cover of $G$, and consider the canonical outcome $\mathbf{o}$ that selects each $c_v$ for $v\in U$ exactly once and selects $z$ in the remaining rounds.
    We verify that $\mathbf{o}$ is $\Phi$-feasible.
    Consider any $\Phi$-constraint induced by some cohesive group $S$ and some $t>0$ for which the demanded amount is positive.
    Since $t\le \ell=n/4$, positivity of $\lfloor t \cdot |S|/n\rfloor$ implies $|S|\ge 4$.
    Since preferences are static and $S$ is cohesive, pick any $x\in\bigcap_{i\in S}A_i$; then $S\subseteq \{i\in N: x\in A_i\}$.
    We distinguish cases.
    
    \begin{itemize}
    \item If $x=z$, then $S\subseteq\{b_1,\dots,b_8\}$, and the maximum possible demand over such groups is
    $\lfloor \ell\cdot 8/n\rfloor = 2$.
    Note that for PJR/EJR, the maximum demand is $2$ (for $|S|=8$); for JR it is at most $1$. Since we select $z$ at least $m+2\ge 2$ times, all these constraints are satisfied.
    
    \item If $x=p_e$ for some $e=\{u,v\}$, then necessarily $S\subseteq \{i\in N: p_e\in A_i\}=S_e$ and $|S|\ge 4$ implies $S=S_e$.
    Since $U$ is a vertex cover, at least one endpoint (say $u$) lies in $U$, so $\mathbf{o}$ selects $c_u$ in some round.
    Because $c_u\in A_{a_{e,4}}$ and $a_{e,4}\in S_e$, this gives the required representation for $S_e$ (and its demand equals $1$ as computed above).
    
    \item If $x=c_v$ for some $v\in V$, then $S\subseteq \{i\in N: c_v\in A_i\}$, which has size $5$; hence the demand is at most $1$ as in Claim~2.
    Moreover, since only $s_{v,1},s_{v,2}$ approve \emph{only} $c_v$, any such $S$ with $|S|\ge 4$ must contain some voter of the form $a_{e,4}$ for an edge $e=\{v,u\}$ incident to $v$.
    If $v\in U$, then $\mathbf{o}$ selects $c_v$ and the constraint is satisfied.
    If $v\notin U$, then because $U$ is a vertex cover we have $u\in U$, so $\mathbf{o}$ selects $c_u$; since $c_u\in A_{a_{e,4}}$, the voter $a_{e,4}\in S$ is satisfied, which meets the $\JR/\PJR/\EJR$ requirement (with demand $1$).
    \end{itemize}
    Thus $\mathbf{o}$ is $\Phi$-feasible.

    If $|U|=k$, then by~\eqref{eq:static-welfare}, a canonical outcome has welfare
    \[
    \util(\mathbf{o})=8(\ell-k)+5k=8\ell-3k.
    \]
    Hence maximizing $\util$ over $\Phi$-feasible outcomes is equivalent to minimizing $k$, and
    \[
    W^\Phi(E_G)=8\ell-3\cdot \mathrm{vc}(G),
    \]
    where $\mathrm{vc}(G)$ is the minimum vertex cover size.
    
    Assume for contradiction that $\Phi\textsc{-MaxUTIL}$ admits a PTAS.
    Given $\varepsilon\in(0,1)$, let $\delta:=\varepsilon/15$ and run the PTAS on $E_G$ to obtain a $\Phi$-feasible outcome $\mathbf{o}$ with
    $\util(\mathbf{o})\ge (1-\delta) \cdot W^\Phi(E_G)$.
    Canonicalize $\mathbf{o}$ using Claims~1--2 (this does not decrease welfare), and let $k$ be the size of the corresponding vertex cover; let $k^*:=\mathrm{vc}(G)$.
    Then
    \[
    8\ell-3k \ \ge\ (1-\delta)(8\ell-3k^*)
    \implies 
    k \ \le\ (1-\delta)k^* + \frac{8\delta\ell}{3}.
    \]
    In a cubic graph, each vertex covers at most $3$ edges, so
    $k^*\ge |E|/3 = (3m/2)/3 = m/2$.
    Thus $m\le 2k^*$, and hence
    \[
    \ell = 2m+2 \le 4k^*+2 \le 6k^*,
    \]
    since $k^*\ge 1$ for cubic graphs.
    Therefore,
    \[
    k \ \le\ (1-\delta)k^* + \frac{8\delta\cdot 6k^*}{3}
    = (1+15\delta)k^*
    = (1+\varepsilon)k^*.
    \]
    So the PTAS for $\Phi\textsc{-MaxUTIL}$ would give us a PTAS for \textsc{Min-Vertex-Cover} on cubic graphs, contradicting APX-completeness unless $\mathrm{P}=\mathrm{NP}$.
    Thus, $\Phi\textsc{-MaxUTIL}$ is APX-hard.
    
    Finally, note that by construction, preferences are static, each voter approves at most $3$ candidates, and each candidate is approved by at most $8$ voters.

\section{Omitted Proofs from Section~\ref{sec:parameterized}}

\subsection{Proof of Theorem~\ref{thm:JR-FPT-static-m}}
    Fix a nonempty voter group $S\subseteq N$ with $\bigcap_{i\in S}A_i\neq\varnothing$ (so $S$ agrees in every round). 
    If $\bigcap_{i\in S}A_i=\varnothing$, then under static preferences $S$ never agrees in any round, hence imposes no JR constraint; thus we consider only $\bigcap_{i\in S}A_i\neq\varnothing$.
    Since the right-hand side in Definition~\ref{defn:jr_pjr_ejr} is nondecreasing in $t$ and $\sat_S(\mathbf{o})$ does not depend on $t$, the strongest JR requirement for such $S$ is obtained by taking $t=\ell$, i.e.,
    \[
    \sat_S(\mathbf{o}) \ge \min\left\{1,\left\lfloor\frac{\ell \cdot |S|}{n}\right\rfloor\right\}.
    \]
    Define $\eta := \left\lceil\frac{n}{\ell}\right\rceil$.
    Then $\min\{1,\lfloor \ell \cdot |S|/n\rfloor\}=1$ if and only if $\ell \cdot |S|\ge n$, equivalently $|S|\ge \eta$; otherwise the JR inequality is trivial. Moreover, under static preferences, for any outcome $\mathbf{o}$ with multiplicities $\mathbf{x}$, we have
    \begin{equation*}
        \sat_S(\mathbf{o})\ge 1
    \iff
    \exists\, i\in S:\ \sat_i(\mathbf{x})\geq 1,
    \end{equation*}
    because $\sat_S(\mathbf{o})$ counts rounds $r$ with $o_r \in \bigcup_{i\in S}A_i$, which happens in some round if and only if some voter in $S$ approves the chosen candidate in that round.
    
    For each candidate $p\in P$ let $N(p):=\{i\in N: p\in A_i\}$ be its set of approvers. Given $\mathbf{x}$, define the set of \emph{unsatisfied} approvers of $p$ by
    \begin{equation*}
        U_p(x) := \{i\in N(p): \sat_i(x)=0\}.
    \end{equation*}
    We claim that a multiplicity vector $\mathbf{x}$ is JR-feasible if and only if
    \begin{equation}\label{eq:JR-candwise}
        |U_p(x)| \leq \eta-1\qquad\text{for all }p\in P.
    \end{equation}
    Indeed, if \eqref{eq:JR-candwise} fails for some $p$, then $U_p(x)$ is cohesive (all its members approve $p$), has size at least $\eta$ (and since $|U_p(\mathbf{x})|\ge \eta=\lceil n/\ell\rceil$, we have $\ell|U_p(\mathbf{x})|\ge \ell\eta\ge n$, and thus $\min\{1,\lfloor \ell|U_p(\mathbf{x})|/n\rfloor\}=1$), and contains no satisfied voter; equivalently $\sat_{U_p(\mathbf{x})}(\mathbf{o})=0$, contradicting JR.
    Conversely, if JR is violated, then there exists a cohesive group $S$ with $|S|\ge \eta$ such that no voter in $S$ is satisfied (so $\sat_i(\mathbf{x})=0$ for all $i\in S$). Pick any $p\in\bigcap_{i\in S}A_i$; then $S\subseteq U_p(x)$, so $|U_p(\mathbf{x})|\ge |S|\ge \eta$, again contradicting \eqref{eq:JR-candwise}.
    
    Let $\mathcal{C}:=\{A_i: i\in N\}$ be the set of distinct approval sets (approval classes). For each class $C\in\mathcal{C}$ let
    $n_C:=|\{i\in N: A_i=C\}|$.
    All voters in class $C$ have the same satisfaction
    $v_C(x):=\sum_{p\in C}x_p$.
    Introduce, for each class $C\in\mathcal{C}$, a binary variable $y_C$ that indicates whether the class is satisfied at least once:
    \[
    y_C=1 \ \Longleftrightarrow\ v_C(x)\ge 1.
    \]
    This is enforced by the linear constraints
    \[
    \sum_{p\in C}x_p \ge y_C, \
    \sum_{p\in C}x_p \le \ell y_C, \
    y_C\in\{0,1\}\text{ for all }C\in\mathcal{C}.
    \]
    Then the number of unsatisfied approvers of a candidate $p$ equals
    \[
    |U_p(x)| = \sum_{C\in\mathcal{C}: p\in C} n_C (1-y_C),
    \]
    and the JR feasibility constraints \eqref{eq:JR-candwise} become the linear inequalities
    \[
    \sum_{C\in\mathcal{C}: p\in C} n_C(1-y_C) \le \eta-1
    \text{ for all } p\in P.
    \]
    
    We therefore solve the following ILP:
    \[
    \max \sum_{p\in P} w(p)\,x_p
    \]
    subject to
    \[
    \sum_{p\in P}x_p=\ell \text{ and } x_p\in\mathbb{Z}_{\ge 0} \text{ for all } p\in P,
    \]
    the class constraints defining the $y_C$ above, and the JR constraints for all $p\in P$.
    By the discussion above, feasible solutions correspond exactly to JR-feasible multiplicity vectors, and the objective equals $\util(\mathbf{x})$.
    
    The ILP has $m$ integer variables $(x_p)_{p\in P}$ and $|\mathcal{C}|$ binary (hence integer) variables $(y_C)_{C\in\mathcal{C}}$. Since $\mathcal{C}\subseteq 2^{P}$, we have $|\mathcal{C}|\le 2^{m}$, so the number of integer variables is bounded by a function of $m$ only. By Lenstra's theorem, this ILP can be solved in time $f(m)\cdot \mathrm{poly}(n,  \ell)$ for some computable function $f$. From an optimal solution $\mathbf{x}$, we output an outcome by selecting each candidate $p$ in exactly $x_p$ rounds (in arbitrary order), which achieves optimal utilitarian welfare among all JR-feasible outcomes.

\subsection{Proof of Theorem~\ref{thm:PJR-FPT-static-m}}
    Fix a nonempty voter group $S\subseteq N$ and define $U_S:=\bigcup_{i\in S}A_i\subseteq P$.
    If $\bigcap_{i\in S}A_i=\varnothing$, then $S$ does not agree in any round (and thus cannot agree in any nonempty subset of rounds), so Definition~\ref{defn:jr_pjr_ejr} imposes no PJR constraint for $S$.
    Otherwise, $S$ agrees in every round, and since $t\mapsto \lfloor t \cdot |S|/n\rfloor$ is nondecreasing, it suffices to enforce the strongest constraint obtained by taking $t=\ell$:
    \[
    \sat_S(\mathbf{o})
    = \bigl|\{r\in[\ell] : o_r\in U_S\}\bigr|
    = \sum_{p\in U_S} x_p
    \ \ge\
    \Bigl\lfloor \frac{\ell \cdot|S|}{n}\Bigr\rfloor .
    \]

    Under static preferences, either $\bigcap_{i\in S}A_i=\varnothing$ and $S$ agrees in no round, or $\bigcap_{i\in S}A_i\neq\varnothing$ and $S$ agrees in every round; hence for cohesive $S$ it suffices to consider $t=\ell$.

    For every subset $U\subseteq P$, define
    \[
    g(U):=
    \begin{cases}
    0, & U=\varnothing,\\[2mm]
    \max\limits_{c\in U}\ \bigl|\{i\in N : c\in A_i \text{ and } A_i\subseteq U\}\bigr|, & U\neq\varnothing.
    \end{cases}
    \]
    We claim that a multiplicity vector $\mathbf{x}$ satisfies PJR if and only if for every $U\subseteq P$,
    \begin{equation}\label{eq:PJR-compressed}
    \sum_{p\in U} x_p \ \ge\ \Bigl\lfloor \frac{\ell\cdot g(U)}{n}\Bigr\rfloor.
    \end{equation}
    
    We first prove the forward direction.
    Fix $U\subseteq P$. If $g(U)=0$, then \eqref{eq:PJR-compressed} is trivial.
    Otherwise, pick $c\in U$ attaining $g(U)$ and let
    \[
    S := \{i\in N : c\in A_i \text{ and } A_i\subseteq U\}.
    \]
    Then $S\neq\varnothing$, and $\{c\} \subseteq \bigcap_{i\in S}A_i$, so $S$ agrees in every round. Moreover $\bigcup_{i\in S}A_i\subseteq U$.
    Applying the (static) PJR inequality to $S$ (with $t=\ell$) gives us
    \[
    \sum_{p\in U}x_p
    \ge 
    \sum_{p\in \cup_{i\in S}A_i}x_p
    \ge
    \Bigl\lfloor \frac{\ell \cdot|S|}{n}\Bigr\rfloor
    =
    \Bigl\lfloor \frac{\ell\cdot g(U)}{n}\Bigr\rfloor,
    \]
    which is exactly \eqref{eq:PJR-compressed}.
    
    Next, we prove the backward direction.
    Let $S\subseteq N$ be any nonempty group with $\bigcap_{i\in S}A_i\neq\varnothing$, and set $U:=U_S=\bigcup_{i\in S}A_i$.
    Pick any $c\in \bigcap_{i\in S}A_i$. Then $c\in U$ and for every $i\in S$ we have $c\in A_i$ and $A_i\subseteq U$, so
    \[
    |S|
    \ \le\
    \bigl|\{i\in N : c\in A_i \text{ and } A_i\subseteq U\}\bigr|
    \ \le\
    g(U).
    \]
    Using \eqref{eq:PJR-compressed} for this $U$ gives
    \[
    \sat_S(\mathbf{o})
    = \sum_{p\in U}x_p
    \ \ge\
    \Bigl\lfloor \frac{\ell\cdot g(U)}{n}\Bigr\rfloor
    \ \ge\
    \Bigl\lfloor \frac{\ell \cdot |S|}{n}\Bigr\rfloor,
    \]
    which is exactly the strongest static PJR requirement for $S$ (and hence implies all weaker requirements for smaller $t$).
    
    By the claim, PJR-feasible outcomes correspond exactly to multiplicity vectors $\mathbf{x}\in\mathbb{Z}_{\ge 0}^{P}$ satisfying $\sum_{p\in P}x_p=\ell$ and \eqref{eq:PJR-compressed} for all $U\subseteq P$.
    Therefore PJR\textsc{-MaxUTIL} reduces to the ILP
    \begin{equation*}
        \max \sum_{p\in P} w(p)x_p
    \end{equation*}
    subject to
    \begin{align*}
        \sum_{p\in P}x_p=\ell, x_p\in\mathbb{Z}_{\ge 0} \text{ for all } p \in P, \text{ and} \\
        \sum_{p\in U}x_p \ge \Bigl\lfloor \frac{\ell\cdot g(U)}{n}\Bigr\rfloor \text{ for all } U\subseteq P.
    \end{align*}
    This ILP has exactly $m$ integer variables. 
    The number of constraints is $2^m+1$, and all coefficients/right-hand sides can be computed from the input in time $f_1(m)\cdot \mathrm{poly}(n,\ell)$.
    By Lenstra's theorem \cite{lenstra1983integer}, an ILP with a fixed number of integer variables can be solved in time $f_2(m)\cdot \mathrm{poly}(n,\ell)$ for a computable function $f_2$ depending only on $m$.
    From an optimal solution $\mathbf{x}$, we output any outcome that selects each candidate $p$ exactly $x_p$ times; under static preferences this outcome achieves welfare $\util(\mathbf{x})$ and satisfies PJR by construction.

\subsection{Proof of Theorem~\ref{thm:fpt-m-candidates-ejr}}
    Fix a nonempty voter set $S\subseteq N$. Under static preferences, $S$ agrees in a round $r$ iff $\bigcap_{i\in S} s_{i,r}=\bigcap_{i\in S} A_i\neq\varnothing$; in that case it agrees in \emph{every} round. Since $t\mapsto \lfloor t|S|/n\rfloor$ is nondecreasing and $\sat_i(\mathbf{x})$ is independent of $t$, it suffices to enforce EJR only for $t=\ell$:
    \begin{equation}\label{eq:thm53-static-ejr}
    \forall\,S\subseteq N\text{ with }\bigcap_{i\in S} A_i\neq\varnothing,\quad
    \exists i\in S: \sat_i(x) \geq \Bigl\lfloor \frac{\ell \cdot |S|}{n}\Bigr\rfloor .
    \end{equation}
    
    For each candidate $c\in P$, let
    \[
    N(c) := \{ i\in N : c\in A_i\}
    \]
    be the set of voters approving $c$. 
    For a fixed $\mathbf{x}$, sort the multiset $\{\sat_i(\mathbf{x}): i\in N(c)\}$ in nondecreasing order and denote the order statistics by
    \[
    b_{c,1}(\mathbf{x}) \le b_{c,2}(\mathbf{x}) \leq \cdots \leq b_{c,|N(c)|}(\mathbf{x}).
    \]
    (If $N(c)=\varnothing$, there are no constraints associated with $c$.)
    
    \begin{claim}\label{clm:thm53-order}
        A multiplicity vector $\mathbf{x}$ satisfies EJR if and only if for every $c\in P$ and every $s\in[|N(c)|]$,
    \begin{equation}\label{eq:thm53-order-ineq}
    b_{c,s}(\mathbf{x}) \ge \Bigl\lfloor \frac{\ell s}{n}\Bigr\rfloor .
    \end{equation}
    \end{claim}
    
    \begin{proof}[Proof of Claim~\ref{clm:thm53-order}]
    We first prove the forward direction.
    Assume \eqref{eq:thm53-static-ejr} holds. Fix $c\in P$ and $s\in[|N(c)|]$, and let $S$ be the set of the $s$ voters in $N(c)$ with the smallest satisfaction values under $\mathbf{x}$. 
    Then $S$ is cohesive because $c\in \bigcap_{i\in S}A_i$, so \eqref{eq:thm53-static-ejr} implies that some $i\in S$ has $\sat_i(\mathbf{x})\ge \lfloor \ell s/n\rfloor$. 
    By construction, $\max_{i\in S}\sat_i(\mathbf{x})=b_{c,s}(\mathbf{x})$, hence \eqref{eq:thm53-order-ineq} follows.
    
    Next, we prove the backward direction.
    Assume \eqref{eq:thm53-order-ineq} holds for all $c$ and $s$. Let $S\subseteq N$ be any nonempty cohesive group and pick some $c\in \bigcap_{i\in S}A_i$; then $S\subseteq N(c)$. 
    Let $s:=|S|$. Among the $|N(c)|$ values $\{\sat_i(\mathbf{x}): i\in N(c)\}$, at most $s-1$ are strictly smaller than $b_{c,s}(\mathbf{x})$.
    Thus, there exists $i\in S$ with $\sat_i(\mathbf{x})\geq b_{c,s}(\mathbf{x})\ge \lfloor \ell s/n\rfloor = \lfloor \ell \cdot |S|/n\rfloor$, which is exactly \eqref{eq:thm53-static-ejr}. 
    Thus $\mathbf{x}$ satisfies EJR.
    \end{proof}
    
    Let $\mathcal{C}:=\{A_i : i\in N\}\subseteq 2^P$ be the set of approval classes, and for each $C\in \mathcal{C}$ let
    \[
    n_C\ := |\{i\in N : A_i=C\}|
    \text{ and }
    v_C(\mathbf{x}) := \sum_{p\in C} x_p.
    \]
    Thus all voters in class $C$ have satisfaction $v_C(x)$. Let $\rho:=|C|$; since $C\subseteq 2^P$, we have $\rho\le 2^m$.
    
    Claim~\ref{clm:thm53-order} reduces EJR to inequalities about the order statistics of satisfactions among $N(c)$ for each $c$. These order statistics depend only on the relative order of the class values $(v_C(\mathbf{x}))_{C\in \mathcal{C}}$, which we linearize by enumerating all possible orders.
    
    Specifically, enumerate all permutations $\pi$ of the $\rho$ classes. 
    For a fixed permutation $\pi$, introduce integer variables $v_C$ (for $C\in \mathcal{C}$) and enforce:
    \begin{equation}\label{eq:thm53-ordpi}
    v_{\pi(1)}\ \le\ v_{\pi(2)} \le \cdots \le v_{\pi(\rho)}.
    \end{equation}
    Ties are allowed; if multiple classes share the same satisfaction value, any order among them is consistent with the sorted multiset.
    For a candidate $c\in P$, let $\pi_c(1),\ldots,\pi_c(k_c)$ be the subsequence of $\pi$ consisting of exactly those classes that contain $c$, in the same relative order as in $\pi$. Define
    \[
    H_{c,0}:=0 \text{ and }
    H_{c,j}:=\sum_{h=1}^j n_{\pi_c(h)}\text{ for } j=1,\ldots,k_c,
    \]
    so that $H_{c,k_c}=|N(c)|$. Under \eqref{eq:thm53-ordpi}, the sorted list of satisfactions of voters in $N(c)$ is obtained by taking $n_{\pi_c(1)}$ copies of $v_{\pi_c(1)}$, then $n_{\pi_c(2)}$ copies of $v_{\pi_c(2)}$, and so on. Hence, for any $s$ with $H_{c,j-1}<s\le H_{c,j}$ we have $b_{c,s}(x)=v_{\pi_c(j)}$.
    
    Because $s\mapsto \lfloor \ell s/n\rfloor$ is nondecreasing, within each block $(H_{c,j-1},H_{c,j}]$ the strongest inequality in \eqref{eq:thm53-order-ineq} occurs at $s=H_{c,j}$. Therefore, for a fixed $\pi$, the family of inequalities \eqref{eq:thm53-order-ineq} for candidate $c$ is equivalent to the boundary constraints
    \begin{equation}\label{eq:thm53-bndpi}
    v_{\pi_c(j)} \ge \Bigl\lfloor \frac{\ell \cdot H_{c,j}}{n}\Bigr\rfloor \text{ for all } c\in P \text{ and } j\in[k_c].
    \end{equation}
    
    Fix a permutation $\pi$. Consider the following ILP with integer variables $(x_p)_{p\in P}$ and $(v_C)_{C\in \mathcal{C}}$:
    \[
    \max \sum_{p\in P} w(p)\,x_p
    \]
    subject to
    \begin{align*}
        \sum_{p\in P} x_p=\ell, x_p\in\mathbb{Z}_{\ge 0}\text{ for all } p\in P, \text{ and } \\
        v_C=\sum_{p\in C} x_p \text{ for all } C\in \mathcal{C},
    \end{align*}
    and the constraints \eqref{eq:thm53-ordpi} and \eqref{eq:thm53-bndpi}.
    Every feasible solution defines a multiplicity vector $\mathbf{x}$ whose induced class satisfactions respect the order $\pi$ and satisfy all inequalities \eqref{eq:thm53-order-ineq}, hence satisfy EJR by Claim~\ref{clm:thm53-order}. Conversely, if $\mathbf{x}$ is EJR-feasible, choose any permutation $\pi$ that orders the classes nondecreasingly by the realized values $v_C(\mathbf{x})$ (breaking ties arbitrarily); then $(\mathbf{x},(v_C(\mathbf{x}))_{C\in \mathcal{C}})$ satisfies \eqref{eq:thm53-ordpi} and \eqref{eq:thm53-bndpi}, so it is feasible for the corresponding ILP and attains the same welfare.
    Indeed, \eqref{eq:thm53-bndpi} follows by applying \eqref{eq:thm53-order-ineq} at $s=H_{c,j}$ and using $b_{c,H_{c,j}}(x) = v_{\pi_c(j)}(x)$ under the consistent order \eqref{eq:thm53-ordpi}.
    
    Thus, the optimal EJR-feasible welfare equals the maximum ILP optimum over all permutations $\pi$. We compute this by enumerating all $\rho!$ permutations and solving each ILP, taking the best solution. Each ILP has $m+\rho\le m+2^m$ integer variables, so by Lenstra's theorem \cite{lenstra1983integer} it can be solved in time $g(m+\rho)\cdot \mathrm{poly}(n,\ell)$ for some computable $g$. Since $\rho\le 2^m$, the total running time is $f(m)\cdot \mathrm{poly}(n,\ell)$ for a computable $f$.
    
    Finally, from an optimal multiplicity vector $x$ we output an outcome $o\in P^\ell$ that selects each candidate $p$ exactly $x_p$ times; this preserves all satisfactions and $\util$, hence gives us an optimal solution to EJR\textsc{-MaxUTIL}.
    
    By Lemma~\ref{lem:static_ejr_ejr+}, under static preferences EJR and EJR+ are equivalent, so the same algorithm also solves EJR+\textsc{-MaxUTIL}.

\subsection{Proof of Theorem~\ref{thm:votertype_jr}}
    For any $U\subseteq T$ write
    \[
    S_U := \bigcup_{\theta\in U} N_\theta \subseteq N.
    \]
    We say that a type set $U\subseteq T$ is \emph{JR-relevant} if there exists an integer $t\ge 1$ such that $S_U$ agrees in a size-$t$ subset of rounds and
    $\lfloor t\cdot |S_U|/n\rfloor \ge 1$ (equivalently, $t\cdot |S_U|\ge n$).

    Let $X(\mathbf{o}) := \{ \theta \in T : \sat_\theta(\mathbf{o}) > 0 \}$ denote the set of satisfied types.    
    Now, we claim that an outcome $\mathbf{o}$ satisfies JR if and only if
    \begin{equation}\label{eq:jr_hit_types}
        \forall\,U\subseteq T\ \text{that are JR-relevant}, X(\mathbf{o})\cap U\neq\varnothing.
    \end{equation}
    
    We first prove the forward direction.
    Fix a JR-relevant $U$ and let $t$ witness relevance.
    Applying JR to the voter group $S_U$ with this $t$ gives us
    \[
    \sat_{S_U}(\mathbf{o}) \geq \min\{1,\lfloor t\cdot |S_U|/n\rfloor\} = 1.
    \]
    By the definition of $\sat_{S_U}(\cdot)$, this implies that some voter in $S_U$ has positive satisfaction, hence some type $\theta\in U$ has $\sat_\theta(\mathbf{o})>0$.
    Therefore $X(\mathbf{o})\cap U\neq\varnothing$.
    
    Next, we prove the backward direction. Assume~\eqref{eq:jr_hit_types}.
    Let $S\subseteq N$ and $t\ge 1$ be such that $S$ agrees in a size-$t$ subset of rounds and $\lfloor t\cdot |S|/n\rfloor\ge 1$.
    Let
    \[
    U := \{\theta\in T: S\cap N_\theta\neq\varnothing\}\subseteq T.
    \]
    Since adding further voters of types already present does not change any per-round intersections (all voters of a type have identical approvals in every round),
    $S_U$ agrees in the same set of rounds as $S$, and thus also agrees in a size-$t$ subset of rounds. Moreover $|S_U|\ge |S|$, so $t\cdot |S_U|\ge n$.
    Thus, $U$ is JR-relevant, and by~\eqref{eq:jr_hit_types} there exists $\theta\in X(\mathbf{o})\cap U$.
    Pick any voter $i\in S\cap N_\theta$; then $i$ has positive satisfaction, so $\sat_S(\mathbf{o})\ge 1$.
    Since $\min\{1,\lfloor t\cdot |S|/n\rfloor\}=1$, the JR inequality for $(S,t)$ holds.
    As $(S,t)$ was arbitrary, $\mathbf{o}$ satisfies JR.

    Now, for each round $j\in[\ell]$ and candidate $p\in P$ define
    \begin{equation*}
        U_j(p):=\{\theta\in T: p\in A_{\theta,j}\}\subseteq T \text{ and } w_j(p):=\sum_{\theta\in U_j(p)} n_\theta.
    \end{equation*}
    Then $\util(\mathbf{o})=\sum_{j=1}^\ell w_j(o_j)$.
    
    For $j\in\{0,1,\dots,\ell\}$ and $X\subseteq T$, let $\mathrm{DP}[j,X]$ be the maximum welfare achievable in the first $j$ rounds by some prefix $(o_1,\dots,o_j)$ such that the set of satisfied types after these $j$ rounds equals $X$.
    Initialize $\mathrm{DP}[0,\varnothing]=0$ and $\mathrm{DP}[0,X]=-\,\infty$ for $X\neq\varnothing$.
    For each $j=1,\dots,\ell$ and each $X\subseteq T$, update using
    \begin{align*}
        & \mathrm{DP}\bigl[j,\,X\cup U_j(p)\bigr] \\
        & = \max\Bigl\{\mathrm{DP}\bigl[j,\,X\cup U_j(p)\bigr],\ \mathrm{DP}[j-1,X]+w_j(p)\Bigr\}
    \end{align*}
    for all $p \in P$.
    Correctness follows by induction on $j$.
    
    To obtain the claimed running time, for each round $j$, it suffices to consider only the distinct type-patterns
    \[
    \mathcal{U}_j:=\{U_j(p): p\in P\}\subseteq 2^T,
    \]
    keeping one representative candidate per pattern (the transition and the value $w_j(\cdot)$ depend only on the pattern).
    Since $|\mathcal{U}_j|\le 2^\kappa$ and there are $2^\kappa$ states $X$, each DP layer is updated in $\mathcal{O}(2^\kappa\cdot 2^\kappa)= \mathcal{O}(4^\kappa)$ time.
    
    Then, recall that an outcome is JR-feasible if and only if its satisfied type set satisfies~\eqref{eq:jr_hit_types}.
    Therefore the optimal JR welfare is equivalent to
    \[
    \max\Bigl\{\mathrm{DP}[\ell,X]:\ X\subseteq T\text{ satisfies }\eqref{eq:jr_hit_types}\Bigr\},
    \]
    and an optimal outcome can be recovered by storing standard backtracking pointers during the DP.

    The DP performs $\mathcal{O}(\ell\cdot 4^\kappa)$ transitions after the
type-pattern families have been computed; including preprocessing, the
running time is $f(\kappa)\cdot \mathrm{poly}(n,m,\ell)$ for
some computable function $f$. Thus, $JR$-MAXUTIL is FPT
with respect to $\kappa$.

\subsection{Proof of Theorem~\ref{thm:ejr-maxutil-types}}
    We claim that an outcome $\mathbf{o}$ satisfies EJR if and only if
    \begin{equation}\label{eq:ejr-type-unions}
    \forall U\subseteq T, U \neq \varnothing \text{ with } d_U>0,
    \max_{\theta\in U} \sat_\theta(\mathbf{o}) \ge d_U .
    \end{equation}
    
    We first show the forward direction.
    Fix a nonempty $U\subseteq T$ with $d_U>0$.
    For any round $r\in[\ell]$, since all voters of a type have identical approval sets in that round,
    intersecting over all voters in $S_U$ is the same as intersecting once per type:
    \[
    \bigcap_{i\in S_U} s_{i,r}\ =\ \bigcap_{\theta\in U} A_{\theta,r}.
    \]
    Hence $S_U$ agrees exactly in the $\gamma_U$ rounds counted by the definition of $\gamma_U$, and therefore $S_U$
    agrees in a size-$\gamma_U$ subset of rounds, so the EJR constraint for $(S_U,t)$ is strongest at $t=\gamma_U$ because $t\mapsto \lfloor t|S_U|/n\rfloor$ is nondecreasing.
    Since EJR demands are monotone in the subset size,
    it suffices to apply EJR to the pair $(S_U,\gamma_U)$ to obtain some voter $i\in S_U$ with
    \begin{equation*}
        \sat_i(\mathbf{o}) \ge \left\lfloor \frac{\gamma_U\cdot |S_U|}{n}\right\rfloor = d_U.
    \end{equation*}
    Let $\theta\in U$ be the type of $i$. Then $\sat_\theta(\mathbf{o})=\sat_i(\mathbf{o})\geq d_U$, proving~\eqref{eq:ejr-type-unions}.
    
    Next, we show the backward direction.
    Assume~\eqref{eq:ejr-type-unions}. Let $S\subseteq N$ be any nonempty voter group and let $t\in\mathbb{N}_{>0}$
    be such that $S$ agrees in a size-$t$ subset of rounds. Define the set of types present in $S$ by
    \[
    U := \{\mathrm{type}(i) : i\in S\}\ \subseteq\ T,
    \]
    so that $S\subseteq S_U$ and $S\cap N_\theta\neq\varnothing$ for every $\theta\in U$.
    For each round $r\in[\ell]$ we again have
    \[
    \bigcap_{i\in S} s_{i,r}\ =\ \bigcap_{\theta\in U} A_{\theta,r},
    \]
    because intersecting multiple identical sets within a type does not change the intersection.
    Therefore, the rounds in which $S$ agrees are \emph{exactly} the rounds counted by $\gamma_U$; in particular,
    $t \le \gamma_U$. Also, since $S\subseteq S_U$, we have $|S|\le |S_U|$, and thus
    \[
    \left\lfloor \frac{t\cdot |S|}{n}\right\rfloor \ \le\
    \left\lfloor \frac{\gamma_U\cdot |S_U|}{n}\right\rfloor \ =\ d_U.
    \]
    If $d_U=0$, then the EJR requirement for the pair $(S,t)$ is trivial. Otherwise $d_U>0$, and~\eqref{eq:ejr-type-unions}
    gives us some $\theta^*\in U$ with $\sat_{\theta^*}(\mathbf{o})\ge d_U$.
    Pick any voter $i^*\in S\cap N_{\theta^*}$ (nonempty by definition of $U$). Then
    \[
    \sat_{i^*}(\mathbf{o}) = \sat_{\theta^*}(\mathbf{o})\ \ge\ d_U\ \ge\ \left\lfloor \frac{t\cdot |S|}{n}\right\rfloor,
    \]
    which is exactly the EJR condition for $(S,t)$. As $(S,t)$ was arbitrary, $\mathbf{o}$ satisfies EJR.
    
    For any outcome $\mathbf{o}$ and type $\theta\in T$, define the truncated value
    \[
    v_\theta(\mathbf{o}) := \min\{D,\sat_\theta(\mathbf{o})\}\ \in\ \{0,1,\ldots,D\}.
    \]
    Because $d_U\le D$ for all $U$, condition~\eqref{eq:ejr-type-unions} is equivalent to
    \begin{equation}\label{eq:ejr-truncated}
    \forall\, U\subseteq T, U \neq \varnothing \text{ with } d_U>0, 
    \max_{\theta\in U} v_\theta(\mathbf{o}) \ge d_U.
    \end{equation}
    
    Now, fix an indexing $T=\{\theta_1,\ldots,\theta_\kappa\}$, and identify a truncated satisfaction profile with a vector
    $v=(v_1,\ldots,v_\kappa)\in\{0,1,\ldots,D\}^{\kappa}$, where $v_j$ corresponds to $\theta_j$.
    
    For each round $r\in[\ell]$ and candidate $p\in P$, define the set of approving types
    \[
    X_r(p) := \{\theta\in T : p\in A_{\theta,r}\}\ \subseteq\ T,
    \]
    and for any $X\subseteq T$ define its welfare weight
    \[
    w(X) := \sum_{\theta\in X} n_\theta.
    \]
    Thus, choosing $p$ in round $r$ contributes exactly $w(X_r(p))$ to utilitarian welfare, and increases
    $\sat_\theta(\cdot)$ by $1$ precisely for $\theta\in X_r(p)$.
    
    For each round $r$, let
    \[
    \mathcal{X}_r := \{ X_r(p) : p\in P \}\ \subseteq 2^T
    \]
    be the family of distinct type approval patterns achievable in that round (so $|\mathcal{X}_r|\le 2^\kappa$).
    For $X\subseteq T$, let $\Delta(X)\in\{0,1\}^{\kappa}$ be the indicator vector with
    \[
    (\Delta(X))_j := \mathbf{1}[\theta_j\in X].
    \]
    
    Define $\mathrm{DP}[r,v]$ to be the maximum utilitarian welfare achievable by an outcome on the first $r$ rounds
    such that the resulting \emph{truncated} satisfaction vector equals $v$.
    Initialize
    \[
    \mathrm{DP}[0,(0,\ldots,0)] = 0 \quad\text{and}\quad \mathrm{DP}[0,v] = -\infty \text{ for all other } v.
    \]
    For each $r=1,\ldots,\ell$, each $v\in\{0,\ldots,D\}^{\kappa}$, and each $X\in\mathcal{X}_r$, define
    \[
    v' := \min\{D,\, v+\Delta(X)\}\quad\text{(componentwise)},
    \]
    and update
    \[
    \mathrm{DP}[r,v'] := \max\bigl\{\mathrm{DP}[r,v'], \mathrm{DP}[r-1,v] + w(X)\bigr\}.
    \]
    For reconstruction, for each $r$ and each $X\in\mathcal{X}_r$ fix a representative candidate $p_{r,X}\in P$
    with $X_r(p_{r,X})=X$, and store standard backtracking pointers for maximizing transitions.
    
    Then, we have shown that an outcome $\mathbf{o}$ is EJR-feasible if and only if its truncated vector $v(\mathbf{o})$ satisfies~\eqref{eq:ejr-truncated}.
    Hence the optimal EJR welfare equals
    \[
    \max\Bigl\{\mathrm{DP}[\ell,v] : v\in\{0,1,\ldots,D\}^{\kappa}\ \text{satisfies }\eqref{eq:ejr-truncated}\Bigr\},
    \]
    and an optimal outcome can be recovered by backtracking and outputting the stored representatives $p_{r,X}$.
    
    There are $(D+1)^{\kappa}$ states per DP layer. In each round $r$ we consider at most $|\mathcal{X}_r|\le 2^{\kappa}$
    actions, and each transition updates $\kappa$ coordinates (which is polynomial in the parameter).
    Thus the DP core runs in time
    $\mathcal{O}(\ell\cdot (D+1)^\kappa\cdot 2^\kappa)$ after preprocessing;
    including the computation of the type-pattern families and the demands
    $(d_U)_{U\subseteq T}$, the running time is
    $f(\kappa,D)\cdot \mathrm{poly}(n,m,\ell)$. Therefore, the algorithm
    is FPT with respect to $\kappa+D$.

\subsection{Proof of Proposition~\ref{prop:votertype_pjr}}
    Let the distinct profiles be indexed by $j\in[q]$, and for each $j$ let $L_j$ denote the number of rounds of profile $j$ (so $\sum_{j=1}^q L_j=\ell$). 
    For $U \subseteq T$, let $S_U := \bigcup_{\theta \in U} N_\theta$. For each profile $j$ and type $\theta$, let $A_{\theta,j}\subseteq P$ for the (common) approval set of type $\theta$ in rounds of profile $j$.
    For $U \subseteq T$, we have $\gamma_U = \sum_{j=1}^q L_j \cdot \mathbf{1}[\bigcap_{\theta \in U} A_{\theta,j} \neq \varnothing]$.
        
    For each profile $j\in[q]$ and candidate $p\in P$, define the set of types approving $p$ under profile $j$ by
    \[
    X_j(p)\ :=\ \{\theta\in T : p\in A_{\theta,j}\}\ \subseteq\ T,
    \]
    and let
    \[
    X_j\ :=\ \{X_j(p): p\in P\}\ \subseteq\ 2^T
    \]
    be the family of achievable type sets in profile $j$.
    
    Fix a nonempty voter group $S\subseteq N$ and let $U:=\{\mathrm{type}(i): i\in S\}\subseteq T$ be the set of types present in $S$.
    Because all voters of the same type have identical approval sets in every round, for each $t\in[\ell]$ we have
    \begin{align*}
        \bigcup_{i\in S} s_{i,t}
    & =\bigcup_{\theta\in U} A_{\theta,t}
    =\bigcup_{i\in S_U} s_{i,t}, \text{ and} \\
    \bigcap_{i\in S} s_{i,t} 
    & =\bigcap_{\theta\in U} A_{\theta,t}
    =\bigcap_{i\in S_U} s_{i,t}.
    \end{align*}
    Hence $S$ and $S_U$ agree in exactly the same set of rounds, so the maximum agreement size for $S$ equals $\gamma_U$.    
    Since $\lfloor t|S|/n\rfloor$ is nondecreasing in $t$, it suffices to enforce PJR only for $t=\gamma_U$, the maximum number of rounds in which the group can agree.
    Moreover, since the unions coincide in every round, $\sat_S(\mathbf{o})=\sat_{S_U}(\mathbf{o})$ for every outcome $o$.
    Finally, $|S|\le |S_U|$ implies
    \[
    \left\lfloor \frac{\gamma_U\cdot |S|}{n}\right\rfloor
    \le
    \left\lfloor \frac{\gamma_U\cdot |S_U|}{n}\right\rfloor
    =d_U.
    \]
    Therefore, if an outcome $\mathbf{o}$ satisfies $\sat_{S_U}(\mathbf{o})\ge d_U$, then it satisfies the strongest PJR requirement for $S$, and thus satisfies all PJR constraints for $S$.
    Conversely, if $\mathbf{o}$ satisfies PJR, then applying PJR to the voter group $S_U$ and $t=\gamma_U$ (the number of rounds in which $S_U$ agrees), we get
    \begin{equation*}
        \sat_{S_U}(\mathbf{o})\geq \left\lfloor \gamma_U\cdot |S_U|/n\right\rfloor = d_U.
    \end{equation*}
    Consequently, an outcome is PJR-feasible if and only if it satisfies $\sat_{S_U}(\mathbf{o})\ge d_U$ for every nonempty $U\subseteq T$.
    
    Fix a profile $j\in[q]$ and $X\in X_j$.
    In a round of profile $j$, selecting a candidate approved by exactly the types in $X$ contributes
    \[
    w(X) := \sum_{\theta\in X} n_\theta
    \]
    to utilitarian welfare.
    Moreover, for any $U\subseteq T$ such a round contributes to $\sat_{S_U}(\mathbf{o})$ if and only if $X\cap U\neq\varnothing$.
    
    Introduce an integer variable $y_{j,X}\in \mathbb{Z}_{\ge 0}$ for every profile $j\in[q]$ and every $X\in X_j$, interpreted as the number of rounds of profile $j$ in which we pick a candidate whose approving type-set is exactly $X$.
    Consider the ILP
    \[
    \max \sum_{j=1}^q\ \sum_{X\in X_j} w(X)\, y_{j,X}
    \]
    subject to
    \begin{align*}
    \sum_{X\in X_j} y_{j,X} & =L_j \text{ for all } j\in[q],\\
    \sum_{j=1}^q\ \sum_{\substack{X\in X_j\\ X\cap U\neq\varnothing}} y_{j,X} & \ge\ d_U \text{ for all } U\subseteq T, U \neq \varnothing, \text{ and}\\
    y_{j,X}\in \mathbb{Z}_{\ge 0} & \text{ for all } j\in[q],\ \forall X\in X_j .
    \end{align*}
    
    We now prove correctness.
    Given an outcome $\mathbf{o}$, define $y_{j,X}$ as the number of rounds $t$ of profile $j$ with $X_j(o_t)=X$.
    Then the profile constraints hold by definition.
    For any nonempty $U\subseteq T$, the left-hand side of the second constraint equals the number of rounds in which the chosen candidate is approved by at least one type in $U$, i.e., it equals $\sat_{S_U}(\mathbf{o})$.
    Thus feasibility of the ILP is equivalent to satisfying $\sat_{S_U}(\mathbf{o})\geq d_U$ for all nonempty $U\subseteq T$, which is equivalent to PJR.
    Finally, the objective equals $\util(\mathbf{o})$ because each round counted in $y_{j,X}$ contributes exactly $w(X)$ to welfare.
    
    Conversely, given a feasible integer solution $(y_{j,X})$, for each $j\in[q]$ and each $X\in X_j$ fix an arbitrary representative candidate $p_{j,X}\in P$ with $X_j(p_{j,X})=X$ (which exists by definition of $X_j$).
    Since candidates may be selected multiple times across rounds, we may use the same representative candidate $p_{j,X}$ in all $y_{j,X}$ rounds of profile $j$ without violating feasibility.
    Assign, for each profile $j$, the candidates $\{p_{j,X}\}_{X\in X_j}$ to its $L_j$ rounds so that $p_{j,X}$ is used in exactly $y_{j,X}$ rounds.
    The resulting outcome satisfies all constraints $\sat_{S_U}(\mathbf{o})\ge d_U$, hence is PJR-feasible, and achieves utilitarian welfare equal to the ILP objective.
    
    For each profile $j$, $|X_j|\le 2^\kappa$, so the number of integer variables is at most
    \[
    p := \sum_{j=1}^q |X_j| \le q\cdot 2^\kappa .
    \]
    The ILP has $q$ profile constraints and at most $2^\kappa-1$ nontrivial PJR constraints.
    All coefficients are integers of magnitude at most $n$ or $\ell$, and each $d_U$ can be computed from the profile counts $(L_j)_{j\in[q]}$ and the sets $(A_{\theta,j})_{\theta\in T,\,j\in[q]}$ in $f_0(\kappa,q)\cdot \mathrm{poly}(n,\ell,m)$ time by enumerating $U\subseteq T$ and using
    \[
    \gamma_U\ =\ \sum_{j=1}^q L_j\cdot \mathbf{1}\left[\bigcap_{\theta\in U} A_{\theta,j}\neq\varnothing\right].
    \]
    By Lenstra's result \cite{lenstra1983integer}, the ILP can be solved in time $g(p)\cdot \mathrm{poly}(n,\ell,m)$ for a computable function $g$.
    Since $p\le q\cdot 2^\kappa$, this is of the form $f(\kappa,q)\cdot \mathrm{poly}(n,\ell,m)$, and an optimal outcome can be reconstructed in time polynomial in the ILP size.

\end{document}